\documentclass[a4paper,12pt]{article}
\usepackage[latin1]{inputenc}
\usepackage{graphicx}
\usepackage[cmex10]{amsmath}
\usepackage[margin=2cm,nohead]{geometry}
\usepackage{cite}
\usepackage{amsfonts}
\usepackage{amssymb}
\usepackage{mathtools}
\usepackage{dsfont}
\usepackage{amsmath}
\usepackage{mathtools}
\usepackage{amsthm}
\usepackage{algorithm}
\usepackage[noend]{algpseudocode}

\usepackage{graphicx}
\usepackage{caption}
\usepackage{subcaption}

\makeatletter
\def\BState{\State\hskip-\ALG@thistlm}
\makeatother

\providecommand{\keywords}[1]{\textbf{\textit{Index terms---}} #1}

\newcommand{\E}[1]{\mathbb{E} \left\{#1 \right\}}

\renewcommand{\P}[1]{\mathbb{P}\left\{ #1 \right\} }
\newcommand{\I}[1]{\mathds{1}\left\{ #1 \right\}}
\newcommand{\Pw}[2]{\mathbb{P}_{#1}\left\{ #2 \right\} }
\newcommand{\bs}[1]{\boldsymbol{#1}}

\newtheorem{definition}{Definition}
\newtheorem{theorem}{Theorem}
\newtheorem{lemma}{Lemma}
\newtheorem{proposition}{Proposition}

\DeclareMathOperator*{\esssup}{ess\,sup}
\DeclareMathOperator*{\essinf}{ess\,inf}

\usepackage{pgfplots}
\pgfplotsset{compat = newest}

\usetikzlibrary{pgfplots.external}

\usepackage{filemod}

\pgfplotscreateplotcyclelist{my black white}{%
dotted, every mark/.append mark=*\\%
dotted, every mark/.append style={solid, fill=gray}, mark=square*\\%
densely dotted, every mark/.append style={solid, fill=gray}, mark=otimes*\\%
loosely dotted, every mark/.append style={solid, fill=gray}, mark=triangle*\\%
dashed, every mark/.append style={solid, fill=gray},mark=diamond*\\%
loosely dashed, every mark/.append style={solid, fill=gray},mark=*\\%
densely dashed, every mark/.append style={solid, fill=gray},mark=square*\\%
dashdotted, every mark/.append style={solid, fill=gray},mark=otimes*\\%
dasdotdotted, every mark/.append style={solid},mark=star\\%
densely dashdotted,every mark/.append style={solid, fill=gray},mark=diamond*\\%
}

\pgfplotscreateplotcyclelist{nano mark style}{%
dashed, every mark/.append style={solid, fill=gray!30}, mark=*, line width=1.1pt\\%
dashed, every mark/.append style={solid, fill=gray!30}, mark=square*, line width=1.1pt\\%
dashed, every mark/.append style={solid, fill=gray!30}, mark=triangle*, mark size=2.5pt, line width=1.1pt\\%
dashed, every mark/.append style={solid, fill=gray!30},mark=diamond*, mark size=2.5pt, line width=1.1pt\\%
dashed, every mark/.append style={solid, fill=gray!30},mark=x, mark size=3pt, line width=1.1pt\\%
}

\pgfplotscreateplotcyclelist{nano line style}{%
color = black, line width=2pt\\%
densely dashed, color = black, line width=2pt\\%
densely dotted, color = black, line width=2pt\\%
loosely dashed, color = black, line width=2pt\\%
}

\begin{document}
\sffamily
\title{Social learning for resilient data fusion\\ against data falsification attacks}
\author{Fernando Rosas$^{1,2}$, Kwang-Cheng Chen$^{3}$ and Deniz G\"{u}nd\"{u}z$^{2}$\vspace{0.2cm}\\
$^1$ \small{Centre of Complexity Science and Department of Mathematics, Imperial College London, UK} \\
$^2$ \small{Department of Electrical and Electronic Engineering, Imperial College London, UK}\\%
$^3$ \small{Department of Electrical Engineering, University of South Florida, USA}\\
}
\date{}
\maketitle

\abstract{
Internet of Things (IoT) suffers from vulnerable sensor nodes, which are likely to endure data falsification attacks following physical or cyber capture. Moreover, centralized decision-making and data fusion schemes commonly used by these networks turn these decision points into single points of failure, which are likely to be exploited by smart attackers. In order to face this serious security thread, we propose a novel scheme that enables distributed data aggregation and decision-making by following social learning principles. Our proposed scheme makes sensor nodes to act resembling the manners of agents within a social network. We analytically examine how local actions of individual agents can propagate through the whole network, affecting the collective behaviour. Finally, we show how social learning can enable network resilience against data falsification attacks, even when a significant portion of the nodes have been compromised by the adversary.
}


\newpage

\section{Introduction}
\label{intro}

\subsection{Motivation}

Internet of Things (IoT) is expected to play a central role in digital society. However, before adopting this technology it is crucial to guarantee its security, specially for those public utilities whose safety is crucial for the well-being of society~\cite{kim2012cyber}. Recent cyber-attacks that created significant damage have been widely reported, e.g., the self-propagating malware \textit{WannaCry} that caused a famous worldwide network hack in May 2017~\cite{wannacry2017}. Therefore, developing technologies for guaranteeing the safety of large information networks such as IoT is a challenging but urgent need. As information networks get more closely intertwined with our daily lifes, ensuring their safety in the future will become an even more challenging issue.

As the level of security is determined by the weakest element, a major dilemma of IoT security lies in the low-complexity sensor networks that are located at the edge of the system. These sensor networks are usually composed by a large number of autonomous electronic devices, which collect information that is critical to the control and operation of IoT~\cite{veeravalli2012distributed,barbarossa2013distributed}. By monitoring extensive geographical areas, these networks enable a wide range of services to civil society, being a key element for the well-being of smart cities \cite{hancke2012role,difallah2013scalable}.
These networks may also perform sensitive tasks, including the surveillance over military or secure zones, intrusion detection to private property, monitoring of drinkable water tanks and protection from chemical attacks~\cite{lambrou2015contamination,lambrou2014low}.

Although the design of secure wireless sensor networks have been widely studied (e.g. \cite{perrig2004security,shi2004designing,pathan2006security} and references therein), there remain open problems of both theoretical and engineering nature~\cite{trappe2015low}. In particular, as the number of sensors is usually very large, a precise management of them is therefore challenging, or even unfeasible. A significant portion of the sensors might be deployed in unprotected areas where it is impossible to ensure their physical or cyber security (e.g., war zones, or regions controlled by an adversary). Furthermore, sensor nodes are generally not tamper-proof due to cost restrictions, and have hence limited computing and networking capabilities. Therefore, they may not be capable of employing reliable cryptographic or security functions of high complexity.

The vulnerability of sensor nodes makes it reasonable to expect that they might be victims of cyber/physical attacks driven by intelligent adversaries. Attacks to information networks are usually categorized into \textit{outside attacks} and \textit{insider attacks}. Outside attacks include (distributed) denial of services (DoS), which use the broadcasting nature for wireless communications to disrupt the communications capabilities~\cite{shi2004designing}. In contrast, in insider attacks the adversary ``recruits'' sensor nodes by malware through cyber/wireless means, or directly by physical substitution~\cite{marano2009distributed}. Following the classic \textit{Byzantine Generals Problem}~\cite{lamport1982byzantine}, these ``Byzantine nodes'' are authenticated, and recognized as valid members of the network. Byzantine nodes can hence generate false data, exhibit arbitrary behaviour, and collude with others to create network malfunctions. In general, inside attacks are considered to be more dangerous to information networks than outside attacks.

The effect of Byzantine nodes and data falsification over distributed sensor networks has been intensely studied; the impact over the network performance has been characterized, and various defense mechanisms has been proposed (c.f. \cite{vempaty2013distributed} for an overview, and also \cite{6778791,6998053,7134807, kailkhura2014distributed, 7065266} for some recent contributions). However, all these works focus on networks with star or tree topology, and rely on centralising the decision-making in special nodes, called ``fusion centers'' (FCs), which gather all the sensed data. Therefore, a key element in these approaches is a strong division of labour: ordinary sensor nodes merely sense and forward data, while the processing is done exclusively at the FC, corresponding to a \textit{distributed-sensing/centralized-processing} approach. This literature implicitly assume that the FCs are capable of executing secure coding and protocols, and hence are out of the reach of attackers. However, large information networks might require another kind of mediator devices, known as data aggregators (DAs), which have the capability to access the cloud through high-bandwidth communication links~\cite{chen2014machine}. DAs are attractive targets to insider attacks, as they might also be located in unsafe locations due to the limited range of sensor nodes' radios. Please note that a tampered DA can completely disable the sensing capabilities of all the nodes whose information has been aggregated, generating a single point of failure that is likely to be exploited by smart adversaries~\cite{parno2005distributed}.

An attractive route to address this issue is to consider \textit{distributed-sensing/distributed-processing} schemes, which avoid centralized decision making while distributing processing tasks throughout the network~\cite{lin2014improving}. However, the design of practical distributed-sensing / distributed-processing schemes is a challenging task, as collective computation phenomena usually exhibit highly non-trivial features~\cite{daniels2016quantifying,brush2018conflicts}. In effect, even though the distributed sensing literature is vast (for classic references c.f. \cite{tsitsiklis1993decentralized,viswanathan1997distributed,blum1997distributed}, and more modern surveys see \cite{chen2006channel,chamberland2007wireless,veeravalli2012distributed,barbarossa2013distributed}), the construction of optimal distributed schemes is in general NP-hard \cite{tsitsiklis1985complexity}. Moreover, although in many scenarios the optimal schemes can be characterized as a set of thresholds for likelihood functions, the determination of these thresholds is usually an intractable problem~\cite{tsitsiklis1993decentralized}. For example, homogeneous thresholds can be suboptimal even for networks with similar sensors arranged in star topology~\cite{warren1999optimum}, being only asymptotically optimal in the network size \cite{chamberland2004asymptotic}. Moreover, symmetric strategies are not suitable for more complicated network topologies, requiring heuristic methods.

\subsection{Distributed decision-making in social learning}

In parallel, significant research efforts have been dedicated to analyzing 
\textit{social learning}, which refers to the decision-making processes that take place within social networks~\cite{easley2010networks}. In these scenarios, agents make decisions based on two elements: private information that represents agent's personal knowledge, and social information derived from previous decisions made by the agent's peers~\cite{acemoglu2011opinion}. 

Social learning was initially investigated by pioneering works that studied sequential decision-making of Bayesian agents over simple social network structures~\cite{banerjee1992simple,bikhchandani1992theory}. These models showed how, thanks to social interactions, individuals with weak private signals can harvest information from the decisions of other agents~\cite{bikhchandani1998learning}. Interestingly, it was also found that aggregation of rational decisions could generate suboptimal collective responses, degrading the ``wisdom of the crowds'' into mere herd behaviour. After these initial findings, researchers have aimed at developing a deeper understanding of information cascades extending the original models by considering more general cost metrics~\cite{smith2000pathological,bala2001conformism,banerjee2004word}, and by studying the effects of the network topology on the aggregated behaviour~\cite{gale2003bayesian,gill2008sequential,acemoglu2011bayesian,hsiao2016}. Non-bayesian learning models have also been explored, where agents use simple rule-of-thumb methods to exchange information~\cite{demarzo2001persuasion,golub2010naive,acemoglu2010spread,jadbabaie2012non,lalitha2014social,rhim2014distributed,7248496}.

Social learning plays a crucial role in many important social phenomena, e.g., in the adoption or rejection of new technology, or in the formation of political opinions~\cite{easley2010networks}. Social learning models are particularly interesting for studying information cascades and herd dynamics, which arises when the social information pushes all the subsequent agents to ignore their own personal knowledge and adopt a homogeneous behaviour~\cite{bikhchandani1992theory}. Moreover, there have been a renewed interest in understanding information cascades in the context of e-commerce and digital society~\cite{hsiao2016}. For example, information cascades might have tremendous consequences in online stores where customers can see the opinions of previous customers before deciding to buy a product, or in the emergence of viral media contents based on sequential actions of like or dislike. Therefore, developing a deep understanding of the mechanics behind information cascades are triggered, and how they impact social learning, is fundamental for our modern networked society.

The main motivation behind this article is to explore the connections that exist between social learning and secure sensor networks, building a bridge between the research done separately by economists and sociologist by one side and electrical engineers and computer scientists by other. A key insight for establishing this link is to realize that each agent's decision correspond to a compressed description of his/her private information. Therefore, the fact that agents cannot access the private information of others, but can only observe their decisions, can be understood as a constraint in the communication resources. In this way, social learning can be regarded as an information network that performs distributed inference (see Table~\ref{table1}). Moreover, it would be natural to use social learning principles in the design of distributed-sensing/distributed-processing schemes, with the hope that this might enables additional robustness to decision-making processes in sensor networks.
\begin{table}[h]
\caption{Table of correspondances}
\label{table1}
\begin{center}
\begin{tabular}{|c|c|} \hline \hline
Distributed detection      & Social learning             \\ \hline\hline
Sensor node                   &  Social agent                \\
Communication range     & Social neighbourhood \\ 
Environmental variables   & State of the world    \\
Noisy measurement         & Private information   \\
Local processing               & Agent's decision         \\
Bandwidth constraints       & Decision sharing        \\ 
\hline\hline
\end{tabular}
\end{center}
\end{table}
%


\subsection{Contribution}

In contrast to almost all the existing research, this work considers powerful topology-aware data falsification attacks, where the adversary knows the network topology and leverages this knowledge to take control of the most critical nodes of the network ---either regular nodes, DAs or FCs.  This represents a worst-case scenario where the network structure has been disclosed, e.g. through network tomography via traffic analysis~\cite{castro2004network}. The reason why this adversary model has not been popular in the literature might be because traditional distributed sensing schemes do not offer any resistance against this kind of attack.

In this work we explore the use of a distributed-sensing/distributed-processing scheme based on social learning principles in order to deal with a topology-aware adversary. The scheme is a threshold-based data fusion strategy, related to those considered in~\cite{tsitsiklis1993decentralized}. However, its relationship with social decision-making allows an intuitive understanding of its mechanisms. For avoiding security threads introduced by fusion centers, our scheme uses tandem or serial decision sequencing~\cite{viswanathan1988optimal,papastavrou1992distributed,swaszek1993performance,viswanathan1997distributed,bahceci2005serial}. It is to be noted that, contrasting with some related literature, our analysis does not focus on optimality aspects of data fusion, but aims to illustrate how distributed decision making can enable network resilience against powerful topology-aware data falsification attacks. We demonstrate how network resilience hold even when a significant number of nodes have been compromised.

Our work exploits a positive effect of information cascades that has been overlooked before: information cascades make a large number of agents/nodes to hold equally qualified estimators, generating many locations where a network operator can collect aggregated data. Therefore, information cascades are crucial in our solution for avoiding single points of failure. For enabling a better understanding of information cascades, this work extends results presented in \cite{rosas2017technological} providing a mathematical characterization of information cascades under data falsification attacks. In particular, our results clarify the conditions upon which local actions of individual agents can propagate across the network, compromising the collective performance. These results provides a first step in the clarification of these non-trivial social dynamics, enriching our understanding of decision-making process in biased social networks. 

This paper expands the ideas presented in \cite{rosas2017social} by developing a formalism that allows considering incomplete or imperfect social information. This formalism is used to overcome the strongest limitation of the scheme presented in \cite{rosas2017social}, namely the fact that each node was required to overhear and store all the previous transmissions in the network. Clearly this cannot take place in a large sensor network, due both to the storage constraints of the nodes, and to the large energy consumption required to transmit and receive across all pairs of nodes~\cite{rosas2012modulation}. Therefore, the present work is an important step towards making this scheme more relevant for practical applications.

The rest of this article is structured as follows. Section~\ref{sec:2} introduces the system model, describing the network controller and the adversary behaviour. Our social learning data fusion scheme is then described in Section~\ref{sec:III}, where some basic statistical properties are explored and practical algorithm for implementing the decision rule is derived. Section~\ref{sec:cascades} analyses the mathematical properties of the decision process, providing a geometrical description and a characterization of information cascades. All these ideas are then illustrated in a concrete scenario in Section~\ref{sec:4}. Finally, Section~\ref{sec:5} summarizes our main conclusions.

\section{System model and problem statement}
\label{sec:2}

This section introduces the basic elements of our framework and settles the basis for our social learning scheme, which is then introduced in Section~3. 
In the rest of the paper uppercase letters $X$ are used to denote random variables and lowercase $x$ realizations of them, while boldface letters $\bs{X}$ and $\bs{x}$ represent vectors. Also, $\Pw{w}{X=x|Y=y} = \P{X=x|Y=y,W=w}$ is used as a shorthand notation.

\subsection{System model}
\label{sec:social_info}

We consider a network of $N$ nodes, each corresponding to an information-processing device that has been deployed in an area of interest. Each node is equipped with a sensor that enables the network to track variables of interest. The measurement of the sensor of the $n$-th node is denoted by $S_n$, taking values over a set $\mathcal{S} \subset \mathbb{R}$ that can be discrete or continuous\footnote{The generalization of our framework and results to vector-valued sensor outputs is straightforward.}
. Based on these signals, the network needs to infer the value of an underlying binary variable $W$.

We consider networks where all the nodes have equal sensing capabilities, that is, the signals $S_n$ are assumed to be identically distributed. Unfortunately, the general distributed detection problem for arbitrarily correlated signals is known to be NP-hard~\cite{tsitsiklis1985complexity}. Hence, for the sake of tractability, it is assumed that the variables $S_1,\dots, S_N$ are conditionally independent given the event $\{W=w\}$\footnote{The conditional independence of sensor signals is satisfied when the sensor noise is due to local causes (e.g. thermal noise), but do not hold when there exist common noise sources (e.g. in the case of distributed acoustic sensors~\cite{bertrand2011applications}). For works that consider sensor interdependence see \cite{kam1992optimal,chen1998adaptive,willett2000good,chamberland2006dense,sundaresan2011copula}.
}, following a probability distribution denoted by $\mu_w$. It is also assumed that both $\mu_0$ and $\mu_1$ are absolutely continuous with respect to each other~\cite{Loeve1978}, i.e., no particular signal determines $W$ unequivocally. This property, in turn, guarantees that the log-likelihood ratio of these two distributions is always well-defined, being given by the logarithm of the corresponding Radon-Nikodym derivative\footnote{When $S_n$ takes a finite number of values then $\frac{d \mu_1}{d \mu_0} (s) = \frac{ \P{ S_n=s|W=1}}{ \P{ S_n=s|W=0}}$, while if $S_n$ is a continuous random variable with conditional p.d.f. $p(S_n|W=w)$ then $\frac{d \mu_1}{d \mu_0} (s) = \frac{ p(s|W=1) }{ p(s|W=0) }$.} $\Lambda_S(s) = \log \frac{d \mu_1}{d \mu_0} (s) $. 

In addition to sensing hardware, each node is equipped with limited computing capability and a wireless radio to transit and receive data. Two nodes in the network are assumed to be connected if they can exchange information wirelessly. Note that sensor nodes usually have a very limited battery budget, which impose severe restrictions on their communication capabilities~\cite{karl2007protocols}. Therefore, it is assumed that each node forwards its data to others only by broadcasting a binary variable $X_n$. These simple signals do not impose an additional burden on the communication resources, as they could be appended to existent wireless control packages and viceversa, or could be shared by light, ultrasound or other alternative media.

We focus on the case in which the sensing capabilities of each sensor are limited; and hence, any inference about $W$ made based only on the sensed data $S_n$ cannot achieve a high accuracy. Interestingly, due to the nature of wireless broadcasting, nearby transmissions can be overheard, and therefore the information that they carry can be fused with what is extracted from the local sensor. The information that a node can extract from overhearing transmissions of other nodes is called ``social information'', contrasting with the ``sensorial information'' that is obtained from the sensed signal $S_n$.

Without loss of generality, nodes transmit their signals sequentially according to their indices (i.e., node 1 transmits first, then node 2, etc.). It is assumed that this sequence is randomly chosen, and can be changed by the network operator at any time (c.f. Section~\ref{sec:at-df}). In general the broadcasted signals $X_1,\dots,X_{n-1}$ might not be directly observable by the $n$-th agent because of various observational restrictions. Therefore, the social observations obtained by the $n$-th node are represented by $\boldsymbol{G}_n\in\mathcal{G}_n$, which can be a random scalar, vector, matrix or other mathematical object. Some cases of interest are:
\begin{itemize}
\item[(i)] The $k$ previous decisions: $\boldsymbol{G}_n = (X_{n-k+1},\dots,X_{n-1})$.
\item[(ii)] The average value of the the previous decisions: $\boldsymbol{G}_n=\frac{1}{n-1} \sum_{k=1}^{n-1} X_k$.
\item[(iii)] The decisions of agents connected by an Erdos-Renyi stochastic network with parameter $\xi\in[0,1]$, i.e. $\boldsymbol{G}_n=(Z_1,\dots,Z_{n-1}) \in\{0,1,e\}^{n-1}$, where
\begin{equation}
Z_k = \begin{cases}
X_k \qquad &\text{with probability }\xi, \\
e    \qquad &\text{with probability } 1-\xi.\end{cases}
\end{equation}
\end{itemize}
Please note that in the last example the Erdos-Renyi model has only been used as an illustrative example, and it can be easily generalized to consider the topology of any stochastic network of interest.

In this work we study the social dynamics based on the properties of the transition probability from states $\bs{g'}\in\mathcal{G}_{n-1}$ to $\bs{g}\in\mathcal{G}_{n}$, as given by the conditional probabilities
\begin{equation}
\beta_w^n(\bs{g} | x_n,\bs{g'}) := \Pw{w}{\bs{G}_n= \bs{g} | X_{n-1}=x_{n-1}, \bs{G}_{n-1}= \bs{g'}},
\end{equation}
where $x_n\in\{0,1\}$. We also assume that the social dynamics is causal, meaning that $\bs{G}_n$ is conditionally independent of $S_m$ given $W$ for all $m\geq n$.

\subsection{The network operator and the adversary}
\label{sec:at-df}

The network is managed by a network operator, who is an external agent that uses the network as a tool to build an estimate of $W$. The network operator is opposed by an adversary, whose goal is to disrupt the inference capabilities of the network. For this aim, the adversary controls a group of authenticated Byzantine nodes without being noticed by the network operator, which have been captured by malware through cyber/wireless means, or by physical substitution.

The overall performance of a network of $N$ nodes is defined by the accuracy of the inference of the $N$-th node, which is the last one in the decision sequence. As the decision sequence is generated randomly by the network operator, every node of the network is equally likely to be at the end of the decision sequence. It is further assumed that the adversary has no knowledge of the decision sequence, as it can be chosen at run-time and changed regularly. 
As the adversary has no reason for capturing any particular node in the network, it is hence reasonable to assume that the adversary captures nodes randomly. Therefore, the Byzantine nodes are considered to be uniformly distributed over the network.

For simplicity, we model the strength of the attack with a single parameter $p_b$, which corresponds to the probability of a node of being compromised\footnote{This attack model assumes implicitly that the capture of each node is an independent event. Extensions considering cyber-infection propagation properties are possible (c.f. \cite{karyotis2016malware}), being left for future studies.}. Moreover, we assume that the capture probability does not depend on $W$\footnote{If the capture ratio would be higher when $W=1$, then detecting Byzantine nodes would provide additional evidence to detect attacks. As this would go against the adversary's interest, we discard this possibility.}. Hence, the number of Byzantine nodes, denoted by $N^*$, is a Binomial random variable with $\E{N^*} = p_bN$. Due to the law of large numbers, $N^*\approx p_bN$ for large networks; and hence, $p_b$ is also the ratio of expected Byzantine nodes in the network, which is the traditional metric for attack strength used in the literature.

For enabling the data processing and forwarding, the network operator defines a \emph{strategy}, i.e. a data fusion scheme given by a collection of (possibly stochastic) functions $\{\pi_n\}_{n=1}^\infty$ such that $\pi_n:\mathcal{S}\times \mathcal{G}_n \to \{0,1\}$ for all $n\in\mathbb{N}$. On the other hand, the adversary can freely set the values of the binary signals transmitted by Byzantine nodes. This is modeled by a random function $C: \{0,1\}\to\{0,1\}$ that corrupts the node's broadcasted signal. Therefore, the broadcasted signal of the $n$-th node is given by
\begin{equation}
X_n = 
\begin{cases}
C(\pi_n(S_n,\boldsymbol{G}_n)) \qquad & \text{with probability }p_b\text{, and} \\
\pi_n(S_n,\boldsymbol{G}_n) \qquad & \text{otherwise.}
\end{cases}
\end{equation}
Furthermore, as broadcasted signals are binary, the corruption function $C(\cdot)$ can be characterized by the conditional probabilities $c_{0|0}$ and $c_{0|1}$, where  $c_{i|j} = \P{ C(\pi) = i | \pi = j }$.

The rest of this work focuses on the case in which the network operator can deduce the corruption function and can estimate the capture risk $p_b$. The average network miss-detection and false alarm rates for an attack of intensity $p_b$ are defined as
\begin{align}
\P{\text{MD};p_b} &:= \Pw{1}{ \pi_N(S_N,\bs{G}_N) = 0 } , \qquad \text{and} \label{eq:MD} \\
\P{\text{FA};p_b}  &:= \Pw{0}{ \pi_N(S_N,\bs{G}_N) = 1 }, \label{eq:FA}
\end{align}
respectively (note that $p_b$ implicitly affects the distribution of $\bs{G}_N$). The case in which these quantities are unknown can be addressed using the current framework with a min-max analysis, which is left for future studies.

\subsection{Problem statement}
\label{qweqw213123}

Our goal is to develop a resilient strategy, in order to provide a reliable estimation of $W$ even under a significant number of unidentified Byzantine nodes. Note that in most surveillance applications miss-detections are more important than false alarms, being difficult to estimate the cost of the worst-case scenario. Therefore, the average network performance is evaluated following the Neyman-Pearson criteria, by setting an allowable false alarm rate $\alpha$ and focusing on reducing the miss-detection rate~\cite{poor2013introduction}. By denoting by $\mathcal{P}$ the set of all strategies, we look to the following optimization problem:
\begin{equation}
\begin{aligned}
&\underset{\{\pi_n\}_{n=1}^\infty \in \mathcal{P}}{\text{minimize}}
& & \P{\text{MD};p_b} \\
& \text{subject to}
& & \P{\text{FP};p_b} \leq \alpha.
\end{aligned}
\end{equation}
However, finding an optimal solution is a formidable challenge, even for the simple case of networks with start topology and no Byzantine attacks~(see \cite{chamberland2007wireless, smith2011network} and references therein). Therefore, our aim is to develop a sub-optimal strategy that enables resilience, while being suitable for implementation in sensor nodes with limited computational power.

\section{Social learning as a data aggregation scheme}
\label{sec:III}

This section describes our proposed data fusion scheme, and explains its functions against topology-aware data falsification attacks. In the sequel, Section~\ref{secA} describes and analyses the data fusion rule, then Section~\ref{sec:statistics} derives basic properties of it statistics, and finally Section~\ref{sec':alg} presents a practical algorithm for its implementation.

\subsection{Data fusion rule}
\label{secA}


Let us assume that each sensor node is a rational agent, who tries to maximizes the profit of an inference within a social network. Rational agents follow \textit{Bayesian strategies}\footnote{Although Bayesian models are elegant and tractable, they assume agents act always rationally~\cite{shiller1995conversation} and make strong assumptions on the knowledge agents have about posterior probabilities~\cite{jadbabaie2012non}. However, Bayesian models provide an important benchmark, not necessarily due to their accuracy but because they give an important reference point with which other models can be compared~\cite{acemoglu2011opinion}.}, which can be elegantly described by the following decision rule~\cite[Ch. 2]{poor2013introduction}:
\begin{equation}\label{eq:bayes}
\frac{ \P{W = 1|S_n,\boldsymbol{G}_{n}} }{ \P{W = 0|S_n,\boldsymbol{G}_{n}} }
\mathop{\lessgtr}_{\pi_n=1}^{\pi_n=0}
\frac{ u(0,0) - u(1,0)  }{ u(1,1) - u(0,1) }
\enspace.
\end{equation}
Above, $u(x,w)$ is a cost assigned to the decision $X_n=x$ when $W=w$, which can be engineered in order to match the relevance of miss-detections and false alarms~\cite{poor2013introduction}. 

Let us find a simpler expression for the decision rule \eqref{eq:bayes}. Due to the causality constrain (c.f. Section~\ref{sec:social_info}), $\bs{G}_n$ can only be influenced by $S_1,\dots,S_{n-1}$, and therefore is conditionally independent of $S_n$ given $W$. Using this conditional independence condition, one can find that
\begin{equation}\label{eq:bayes2}
\frac{ \P{W = 1|S_n,\boldsymbol{G}_{n}} }{ \P{W = 0|S_n,\boldsymbol{G}_{n}} } 
= e^{\Lambda_S(S_n) + \Lambda_{\boldsymbol{G}_{n}}(\boldsymbol{G}_{n})},
\end{equation}
where $\Lambda_{\boldsymbol{G}_{n}}(\boldsymbol{G}_{n}) $ is the log-likelihood ratio of $\boldsymbol{G}_{n}$. Then, using \eqref{eq:bayes2} one can re-write \eqref{eq:bayes} as
\begin{equation}\label{eq:asdasdA3}
\Lambda_S(S_n) + \Lambda_{\boldsymbol{G}_n}(\boldsymbol{G}_n) \mathop{\lessgtr}_{\pi_n=1}^{\pi_n=0} \tau_0
\enspace,
\end{equation}
where $\tau_0 =  \log \frac{ \P{W=0} }{ \P{W=1} } + \log \frac{ u(0,0) - u(1,0) }{ u(1,1) - u(0,1)  }$. In simple words, \eqref{eq:asdasdA3} states how the the $n$-th node should fuse the private and social knowledge: the evidence is provided by the corresponding log-likelihood terms, which are then simply added and then compared against a fixed threshold\footnote{As the prior distribution of $W$ is usually unknown, $\tau_0$ is a free parameter of the scheme. Following the discussion in Section~\ref{qweqw213123}, the network operator shall select the lowest value of $\tau$ that satisfies the required false alarm rate given by the Neyman-Pearson criteria.}.

Further understanding of the above decision rule can be attained by studying it from the point of view of communication theory~\cite{rosas2017technological}. We first note that the decision is made not over the full raw signal $S_n$ but over the ``decision signal'' $\Lambda_S(S_n)$, which is a processed version of it. Interestingly, this processing might serve for dimensionality reduction, as even though $S_n$ can be a matrix or a high-dimensional vector $\Lambda_S(S_n)$ is always a single number. Due to their construction and the underlying assumptions over $S_n$ (c.f. Section~\ref{sec:social_info}), the variables $\Lambda_S(S_n)$ are identically distributed and conditionally independent given $W=w$. Moreover, by introducing the shorthand notation $\tau_n (\boldsymbol{G}_n) = \tau_0 - \Lambda_{\boldsymbol{G}_n}(\boldsymbol{G}_n)$, one can re-write \eqref{eq:asdasdA3} as
\begin{equation}\label{eq:asdasdA3}
\Lambda_S(S_n) \mathop{\lessgtr}_{\pi_n=1}^{\pi_n=0} \tau_n(\boldsymbol{G}_n)
\enspace.
\end{equation}
Therefore, the decision is made by comparing the decision signal with a decision threshold $\tau_n$. Note that this represents a comparison between the sensed data, summarised by $\Lambda_S(S_n)$, and the social information carried by $\tau_n(\boldsymbol{G}_n)$.

\subsection{Decision statistics}\label{sec:statistics}

Let us find expressions for the probabilities of the actions of the $n$-th agent, first focusing on the case $n=1$. Note that
\begin{equation}
\Pw{w}{\pi_1(S_1) = 0 } = \Pw{w}{ \Lambda_S(S_1) < \tau_0 } = F_w^\Lambda(\tau_0)
\end{equation}
where $F_w^\Lambda(\cdot)$ is the c.d.f. of $\Lambda_S$ conditioned on $W=w$. 
Then, considering the possibility that the first node could be a Byzantine, one can show that
\begin{align}\label{eq:asqwt}
\Pw{w}{ X_1 =0 } &= p_b \Pw{w}{ X_1 = 0|\text{ \small Byzantine}} + (1-p_b)\Pw{w}{X_1=0 | \text{\small not a Byzantine} } \nonumber \\
&= p_b( c_{0|0} F_w^\Lambda(\tau_0) + c_{0|1} [ 1 - F_w^\Lambda(\tau_0) ] )  + (1-p_b) F_w^\Lambda(\tau_0) \\
&=z_0 + z_1 F_w^\Lambda(\tau_0)
\enspace,
\end{align}
where we are introducing $z_0:= p_bc_{0|1}$ and $z_1:= 1 - p_b(1-c_{0|0} + c_{0|1} )$ as short-hand notation, which are non-negative constants that summarize the strength of the adversary. In particular, when the adversary is powerless then $z_0=0$ and $z_1 = 1$ and hence $\Pw{w}{\pi_1(S_1)=0} = \Pw{w}{X_1=0}$. 

By considering the $n$-th node, one can find that
\begin{align}
\Pw{w}{\pi_n(S_n,\boldsymbol{G}_n) = 0 |\boldsymbol{G}_n=\boldsymbol{g}_n} &= \int_\mathcal{S} 
\Pw{w}{ \pi_n(s_n,\boldsymbol{g}_n) = 0 | S_n=s } \mu_w(s) \text{d} s  \nonumber\\
&= \int_\mathcal{S} \I{ \pi_{n} (\boldsymbol{g}_{n}, s) = 0 } \mu_w(s) \text{d} s \\
&= \Pw{w}{ \Lambda_S(s) < \tau_{n}(\boldsymbol{g}_{n}) }\\
&= F_w^\Lambda(\tau_{n}(\boldsymbol{g}_{n}))  \label{eq:tauuuw}
\enspace.
\end{align}
The first equality is a consequence of the fact that $S_n$ is conditionally independent of $\boldsymbol{G}_n$ given $W=w$, while the second equality is a consequence that $X_n$ can be expressed as a deterministic function of $\boldsymbol{G}_{n}$ and $S_n$, and hence becomes conditionally independent of $W$. Note that \eqref{eq:tauuuw} shows that $\tau_n$ is a sufficient statistic for predicting $X_n$ with respect to $\boldsymbol{G}_{n}$. Note that $F_w^\Lambda(x)$ can be directly computed from the statistics of the signal distribution (its properties are explored in Appendix~\ref{sec:appendix_sig}). Moreover, using \eqref{eq:tauuuw} and following a similar derivation as in \eqref{eq:asqwt}, one can conclude that 
\begin{equation}\label{adasdh5455}
\Pw{w}{ X_{n} = 0 | \boldsymbol{G}_{n}= \boldsymbol{g}_{n}} = z_0 + z_1 F_w^\Lambda(\tau_n(\boldsymbol{g}_n) ) .
\end{equation}

Let us now study the statistics of $\boldsymbol{G}_n$. 
By using the definition of the transition coefficients $\beta_w^n(\boldsymbol{g}_{n+1}|x_n,\boldsymbol{g}_{n})$, one can find that
\begin{equation}\label{qwe"123123}
\Pw{w}{\boldsymbol{G}_{n+1}=\boldsymbol{g}_{n+1}} = \sum_{\boldsymbol{g}_{n} \in \mathcal{G}_{n}} \sum_{x_n\in\{0,1\}} \beta_w^n(\boldsymbol{g}_{n+1}|x_{n},\boldsymbol{g}_{n}) \Pw{w}{ X_{n} = x_n, \boldsymbol{G}_{n} = \boldsymbol{g}_{n}}.
\end{equation}
Note that, using the above derivations, the terms $\Pw{w}{ X_{n} = x_n, \boldsymbol{G}_{n} = \boldsymbol{g}_{n}}$  can be further expressed as
\begin{align}
\Pw{w}{ X_{n} = x_{n}, \boldsymbol{G}_{n}=\boldsymbol{g}_n} &= \Pw{w}{ X_n=x_{n}| \boldsymbol{G}_n=\boldsymbol{g}_n} \Pw{w}{ \boldsymbol{G}_n=\boldsymbol{g}_n } \\
&= \lambda( z_0 + z_1 F_w^\Lambda(\tau_n(\boldsymbol{g}_n)),x_n) \Pw{w}{ \boldsymbol{G}_n=\boldsymbol{g}_n },
\end{align}
where $\lambda(p,x) = x (1-p) + (1-x) p$. Therefore, a closed form expression can be found for \eqref{qwe"123123} recursively over $\boldsymbol{G}_n$.

\subsection{An algorithm for computing the social log-likelihood}
\label{sec':alg}

The main challenge for implementing \eqref{eq:asdasdA3} as a data processing method in a sensor node is to have an efficient algorithm for computing $\tau_n(\boldsymbol{g}_n)$. Leveraging the above derivations, we develop Algorithm 1 as a iterative procedure for computing $\tau_n$. 
\begin{algorithm*}
\caption{Computation of the decision threshold}\label{alg:social}
\begin{algorithmic}[1]
\Function{Compute\_Tau}{$N,F_0^\Lambda(\cdot),F_1^\Lambda(\cdot),\beta_w^n(\cdot|\cdot,\cdot), \tau_0, z_0,z_1$}
\State $\tau_1 = \tau_0$
\For{$x_1 \in \{0,1\}$}
\State $\Pw{0}{X_1=x_1,\boldsymbol{G}_1=0} = \lambda(z_0 + z_1 F_0^\Lambda(\tau_1),x_1)$
\State $\Pw{1}{X_1=x_1,\boldsymbol{G}_1=0} = \lambda(z_0 + z_1 F_1^\Lambda(\tau_1),x_1)$
\EndFor
\For{$n=1,\dots,N-1$}
\For{$\forall \boldsymbol{g} \in\mathcal{G}_{n+1}$}
\State $\Pw{0}{\boldsymbol{G}_{n+1}=\boldsymbol{g}} = \sum_{\boldsymbol{g}_{n} \in \mathcal{G}_{n}}\sum_{x_{n}=\{0,1\}} \beta_w^n(\boldsymbol{g}_{n+1}|x_{n},\boldsymbol{g}_{n}) \Pw{0}{X_{n}=x_n,\boldsymbol{G}_{n} = \boldsymbol{g}_{n}}$
\State $\Pw{1}{\boldsymbol{G}_{n+1}=\boldsymbol{g}} = \sum_{\boldsymbol{g}_{n} \in \mathcal{G}_{n}}\sum_{x_{n}=\{0,1\}} \beta_w^n(\boldsymbol{g}_{n+1}|x_{n},\boldsymbol{g}_{n}) \Pw{1}{X_{n}=x_n,\boldsymbol{G}_{n} = \boldsymbol{g}_{n}}$
\State $\Lambda_{\boldsymbol{G}_n}(\boldsymbol{g}) = \log \frac{\Pw{1}{\boldsymbol{G}_n=\boldsymbol{g}}}{\Pw{0}{\boldsymbol{G}_n=\boldsymbol{g}}} $
\State $\tau_n(\boldsymbol{g}) = \nu+\eta - \Lambda_{\boldsymbol{G}_n}(\boldsymbol{g})$
\For{$x_{n+1} \in \{0,1\}$}
\State $\Pw{0}{ X_{n+1} = x_{n+1} , \boldsymbol{G}_{n+1}  = \boldsymbol{g} } = \lambda(z_0+z_1 F_0^\Lambda(\tau_n(\boldsymbol{g}_n)),x_{n+1}) \Pw{0}{\boldsymbol{G}_{n+1}  = \boldsymbol{g} } $
\State $\Pw{1}{ X_{n+1} = x_{n+1} , \boldsymbol{G}_{n+1}  = \boldsymbol{g} } = \lambda(z_0 + z_1F_1^\Lambda(\tau_n(\boldsymbol{g}_n)),x_{n+1}) \Pw{1}{\boldsymbol{G}_{n+1}  = \boldsymbol{g} } $
\EndFor
\EndFor
\EndFor
\State \Return $\tau_N(\cdot)$
\EndFunction
\end{algorithmic}
\end{algorithm*}

In many cases of interest the algorithm's complexity scales gracefully. For the particular case of nodes with memory of length $k$ (i.e. $\boldsymbol{G}_n=(X_{n-k-1},\dots,X_{n-1})$), the algorithmic complexity of Algorithm~1 is $\mathcal{O}( 2^k N)$, and therefore grows linearly with the size of the network, while being limited in the values of $k$ that can consider. In general, the algorithm complexity scales linearly with $N$ as long as the cardinality of $\mathcal{G}_n$ are bounded, or if a significant portion of the terms $\beta_w^n(\boldsymbol{g}_{n+1} | x_n,\boldsymbol{g}_n)$ are zero.

The inputs that drive Algorithm~1 can be classified in two groups. First, the terms $N,F_0^\Lambda(\cdot),F_1^\Lambda(\cdot),\beta_w^n(\cdot|\cdot,\cdot)$ are properties of the network (position of the node within the decision sequence, sensor statistics and social observability, respectively) that the network operator could measure. On the other hand, $\tau_0, z_0,z_1$ are properties of the adversary profile that depend on the prior statistics of $W$, $p_b$ and the corruption function defined by $c_{0|0}$ and $c_{0|1}$ (c.f. Section~\ref{sec:at-df}). In most scenarios the knowledge of the network controller about these quantities is limited, as attacks are rare and might follow unpredictable patters. Limited knowledge can still be exploited using e.g. Bayesian estimation techniques~\cite{gelman2014bayesian}. If no knowledge is available for the network controller, then these quantities can be considered free parameters of the strategy that span a range of alternative balances between miss-detections and false positives, i.e. a receiver operating characteristic (ROC) space.

\section{Information cascade}
\label{sec:cascades}

The term ``social learning'' refers to the fact that $\pi_n(S_n,\boldsymbol{G}_n)$ becomes a better predictor of $W$ as $n$ grows, and hence larger networks tend to develop a more accurate inference. However, as the number of shared signals grows, the corresponding ``social pressure'' can make nodes to ignore their individual measurements to blindly follow the dominant choice, triggering a cascade of homogeneous behaviour. It is our interest to clarify the role of the social pressure in the decision making of the agents involved in a social network, as information cascades can introduce severe limitations in the asymptotic performance of social learning~\cite{acemoglu2011bayesian}.

Moreover, an adversary can leverage the information cascade phenomenon. In effect, if the number of Byzantine nodes $N^*$ is large enough then a misleading information cascade can be triggered almost surely, making the learning process to fail. However, if $N^*$ is not enough then the network may undo the pool of wrong opinions and end up triggering a correct cascade.

In the sequel, the effect of information cascades is first studied in individual nodes in Section~\ref{sec:local}. Then, the propagation properties of cascades is explored in Section~\ref{sec:5a}.

\subsection{Local information cascades}
\label{sec:local}

In general, the decision $\pi_n(S_n,\boldsymbol{G}_n)$ is made based in the evidence provided by both $S_n$ and $\boldsymbol{G}_{n}$. A \emph{local cascade} takes place in the $n$-th agent when the information conveyed by $S_n$ is ignored in the decision-making process due to a dominant influence of $\boldsymbol{G}_n$. We use the term ``local'' to emphasize that this event is related to the data fusion of an individual agent. This idea is formalized in the following definition using the notion of conditional mutual information ~\cite{cover2012elements}, denoted as $I(\cdot;\cdot|\cdot)$.

\begin{definition}
\label{local}
The social information $\boldsymbol{g}_{n} \in \mathcal{G}_n$ generates a \emph{local information cascade} for the $n$-th agent if $I(\pi_n;S_n|\boldsymbol{G}_n = \boldsymbol{g}_n) = 0$.
\end{definition}

The above condition summarizes two possibilities: either $pi_n$ is a deterministic function of $\boldsymbol{G}_n$ and hence there is no variability in $\pi_n$ after $\boldsymbol{G}_n$ has been determined, or there is still variability (i.e. $\pi_n$ is a stochastic strategy) but it is conditionally independent of $S_n$. In both cases, the above formulation highlights the fact that the decision $\pi_n$ contains no information coming from $S_n$\footnote{Recall that $S_n$ and $\boldsymbol{G}_n$ are conditionally independent given $W=w$ (c.f. Section~\ref{secA}), and hence there cannot be redundant information about $W$ that is conveyed by $S_n$ and also $\boldsymbol{G}_n$. For a more detailed discussion about redundant information c.f. \cite{rosas2016understanding}.}. 

\begin{lemma}
The variables $\boldsymbol{G}_n \rightarrow \tau_n \rightarrow \pi_n$ form a Markov Chain (i.e. $\tau_n$ is a sufficient statistic of $\boldsymbol{G}_n$ for predicting the decision $\pi_n$).
\end{lemma}
\begin{proof}
Using \eqref{eq:tauuuw} one can find that
$$
 \Pw{w}{\pi_n|\tau_n,\boldsymbol{G}_n} = \lambda (F_w^\Lambda(\tau_n),X_n) =  \Pw{w}{\pi_n|\tau_n},
$$
and therefore the conditional independency of $\pi_n$ and $\boldsymbol{G}_n$ given $\tau_n$ is clear.
\end{proof}

Let us now introduce the notation $U_s = \esssup_{s\in\mathcal{S}} \Lambda_S(S_n=s)$ and $L_s = \essinf_{s\in\mathcal{S}} \Lambda_S(S_n=s)$ for the essential supermum and infimum of $\Lambda_S(S_n)$, being the signals within $\mathcal{S}$ that most strongly support the hypothesis $\{W=1\}$ over $\{W=0\}$ and viceversa\footnote{The essential supremum is the smallest upper bound over $\Lambda_S(S_n)$ that holds almost surely, being the natural measure-theoretic extension of the notion of supremum \cite{dieudonne1976}.}. If one of these quantities diverge, this would imply that there are signals $s\in\mathcal{S}$ that provide overwhelming evidence in favour of one of the competing hypotheses. If both are finite then the agents are said to have \textit{bounded beliefs}~\cite{acemoglu2011bayesian}. As sensory signals of electronic devices are ultimately processed digitally, the number of different signals that an agent can obtain are finite and hence their supremum is always finite. Therefore, in the sequel we asume that both $L_s$ and $U_s$ are finite. Using these notions, the following proposition provides a characterization for local information cascades.

\begin{proposition}\label{prop:local}
The social information $\boldsymbol{g}_{n} \in \mathcal{G}_n$ triggers a local information cascade if and only if the agents have bounded beliefs and $\tau_n(\boldsymbol{g}_{n}) \notin [L_s,U_s]$.
\end{proposition}

\begin{proof}
Let us assume that the agents have bounded beliefs. From the definition of $F_w^\Lambda$, which is a cumulative density function, it is clear that if $\tau_n<L_s$ then $F_0^\Lambda(\tau_n) = F_1^\Lambda(\tau_n) = 0$, while if $\tau_n>U_s$ then $F_0^\Lambda(\tau_n) = F_1^\Lambda(\tau_n) = 1$. Therefore, if $\tau_n(\boldsymbol{g}_{n}) \notin[L_s,U_s]$ then, according to \eqref{eq:tauuuw}, it determines $\pi_n$ almost surely, making $\pi_n$ and $S_n$ conditionally independent. 

To prove the converse by contrapositive, let us assume that $L_s < \tau_n(\boldsymbol{g}_{n}) < U_s$. Using again \eqref{eq:tauuuw} and the definition of $U_s$ and $L_s$, one can conclude that this implies that $0 < \Pw{w}{\pi_n=0|\boldsymbol{G}_n } < 1$ for both $w\in\{0,1\}$. This, in turn, implies that the sets $\mathcal{S}^0(\tau) = \{ s\in\mathcal{S} | \Lambda_S(s) < \tau_n(\boldsymbol{G}_n \}$ and $\mathcal{S}^1(\tau) = \mathcal{S} - \mathcal{S}^0$ both have positive probability under $\mu_0$ and $\mu_1$, which in turn implies the existence of conditional interdependency between $\pi_n$ and $S_n$ in this case.
\end{proof}

Intutively, Proposition~\ref{prop:local} shows that a local information cascade happens when the social information goes above the most informative signal that could be sensed. Some consequences of this result are explored in the next section.

\subsection{Social information dynamics and global cascades}
\label{sec:5a}

It is of great interest to predict when a local information cascade could propagate across the network, disrupting the collective behaviour and hence affecting the network performance. The following definition captures how, during a ``global information cascade'', the shared signals $X_n$ do not convey anymore information from the corresponding sensor signals. 

\begin{definition}
\label{global}
The social information $\boldsymbol{g}_n\in\mathcal{G}_n$ triggers a \emph{global information cascade} if $I(X_m:S_m|\boldsymbol{G}_n = \boldsymbol{g}_n) = 0$ holds for all $m\geq n$.
\end{definition}

A global information cascade is a succession of local information cascades. As Proposition~\ref{prop:local} showed that agents are free from local cascades as long as $\tau_n\in [L_s,U_s]$, one can guess that global cascades are related to the dynamics of $\tau_n$. These dynamics are determined by the transitions of $\boldsymbol{G}_n$, which follows the behaviour dictated by the transition coefficients $\beta_w^n(\cdot|\cdot,\cdot)$. To further study the social information dynamics we introduce the following definitions.
\begin{definition} The collection $\{\boldsymbol{G}_n\}_{n=1}^\infty$ is said to have:
\begin{itemize}
\item[1.] \emph{strongly consistent transitions} if, for any $W=w$, $\boldsymbol{g}\in\mathcal{G}_n$ and $\boldsymbol{g'}\in\mathcal{G}_{n-1}$, $\beta_w^n( \boldsymbol{g}|1,\boldsymbol{g'} )>0$ implies $\tau_{n}(\boldsymbol{g}) \leq \tau_{n-1}(\boldsymbol{g'})$, while if $\beta_w^n(\boldsymbol{g}|0,\boldsymbol{g'})>0$ implies $\tau_{n}(\boldsymbol{g}) \geq \tau_{n-1}(\boldsymbol{g'})$. 
\item[2.] \emph{weakly consistent transitions} if, for all $\boldsymbol{g}\in\mathcal{G}_n$ and $\boldsymbol{g'}\in\mathcal{G}_{n-1}$, $\tau_{n-1}(\boldsymbol{g'}) \leq L_s$ and $\Pw{w}{\boldsymbol{G}_n=g|\boldsymbol{G}_{n-1}=\boldsymbol{g'}} >0$ implies $\tau_{n}(\boldsymbol{g}) \leq L_s$, while $\tau_{n-1}(\boldsymbol{g'}) \geq U_s$ and $\Pw{w}{\boldsymbol{G}_n=\boldsymbol{g}|\boldsymbol{G}_{n-1}=\boldsymbol{g'}} >0$ implies $\tau_{n}(\boldsymbol{g}) \geq U_s$ \footnote{Note that the condition $\Pw{w}{\boldsymbol{G}_n=\boldsymbol{g}|\boldsymbol{G}_{n-1}=\boldsymbol{g'}} >0$ is equivalent to ask either $\beta_w^n(\boldsymbol{g},|0, \boldsymbol{g'})$ or $\beta_w^n(\boldsymbol{g},|1, \boldsymbol{g'})$ to be strictly positive.}.
\end{itemize}
\end{definition}

Intuitively, strong consistency mean that the decision threshold evolves monotonically with respect to the broadcasted signals $X_n$. Correspondingly, weak consistency implies that $\tau_n$ cannot return into $[L_S,U_S]$ after going beyond it. Moreover, the adjectives ``strong'' and ``weak'' reflect the fact that weak consistency only takes place outside the boundaries of the signal likelihood, while the strong consistency affects all the decision space. Moreover, strong consistent transitions imply weak consisten transitions when there are no Byzantine nodes, as shown in the next Lemma\footnote{It is possible to build examples where weak consistency does not follow strong consistency when $p_b>0$.}. 

\begin{lemma} \label{lemma12123123}
Strong consistent transitions satisfy the weak consistency condition if $p_b=0$.
\end{lemma}
\begin{proof}
See Appendix~\ref{app:2}.
\end{proof}

We show next that if the evolution of $\bs{G}_n$ becomes deterministic and 1-1 after leaving the interval $[L_s,U_s]$ (henceforth called \textit{weakly invertible transitions}), then it satisfies the weak consistency condition. 
\begin{lemma}\label{ffewafawef}
Weekly invertible transitions imply the weakly consistency condition. 
\end{lemma}
\begin{proof}
See Appendix~\ref{app:3}.
\end{proof}

Now we present the main result of this section, which is the characterization of information cascades for the case of social information that follows weakly consistent transitions.

\begin{theorem}\label{teo2}
If the social information have weakly consistent transitions, then a global information cascade is triggered by each local information cascade.
\end{theorem}

\begin{proof}
Let us consider $\bs{g}_0\in\mathcal{G}_n$ such that it produce a local cascade in the $n$-th node. Then, due to Proposition~\ref{prop:local}, this implies that $\tau_n(\bs{g})\notin [L_s,U_s]$ almost surely. This, combined with the weakly consistency assumption, implies that $\tau_{n+1}(\bs{G}_{n+1})\notin [L_s,U_s]$ almost surely. A second application of Proposition~\ref{prop:local} make us to conclude that $\Pw{w}{\pi = 0 | \bs{G}_{n+1}}$ is therefore equal to $0$ o $1$. This, in turn, guarantees that $I(\pi_{n+1}:S_{n+1} | \bs{G}_{n} = \bs{g}) = 0)$ almost surely, showing that the $(n+1)$-th node experiences a local information cascade because of $\bs{G}_n = \bs{g}_0$. 

Finally, it is direct to see that a recursive application of the above argument allows one to prove that  $I(\pi_{n+m}:S_{n+m} | \bs{G}_{n} = \bs{g}) = 0)$ for all $m\geq 0$, confirming the existence of a global cascade.
\end{proof}

This theorem has a number of important consequences. Firstly, it provides an intuitive geometrical description about the nature of global cascades for networks with weak consistency. One can imagine the evolution of $\tau_n(\bs{G}_n)$ as function of $n$ as a random walk within the interval $[L_s,U_s]$. Because of the weakly consistency condition, if the random walk step out of the interval, it will never come back. Moreover, as consequence of this theorem, the stepping out of $[L_s,U_s]$ is a necessary and sufficient condition to trigger a global information cascade over the network.

Also, note that in the case where $G_n = \boldsymbol{X}^n$ (i.e. each node see every previous decision) is direct to prove that $G_n$ has weakly invertible transitions. Therefore, Theorem~\ref{teo2} is a generalization of Theorem~1 of \cite{rosas2017technological}, now being valid for the case where there are a fraction of Bizantine nodes within the network.

\section{Proof of concept}
\label{sec:4}

This section illustrates the main results obtained in Sections~\ref{sec:III} and \ref{sec:cascades} in a simple scenario. In the sequel, first Section~\ref{sec:scenario} describes the scenario, and then Section~\ref{sec:discussion} discusses numerical simulations.

\subsection{Scenario description}
\label{sec:scenario}

Let us consider a sensor network that has surveillance duties over a sensitive geographical area. The sensitive area could correspond to a factory, a drinkable water container or a warzone, whose key variables need to be supervised. The task of the sensor network is, through the observation of these variables, to detect the events $\{W=1\}$ and $\{W=0\}$ that correspond to the presence or absence of an attack to the surveilled area, respectively. No knowledge about of the prior distribution of $W$ is assumed.

We consider nodes that have been deployed randomly over the sensitive area, and hence their location follow a Poisson Point process (PPP). The ratio of the area of interest that falls within the range of each sensor is denoted by $r$. If attacks occur uniformly over the surveilled area, then $r$ is also the probability of an attack taking place under the coverage area of a particular sensor is. Note that, due to the limited sensing range, the miss-detection rate of individual nodes is roughly equal to $1-r$. As $r$ is usually a small number ($5\%$ in our simulations), this implies that without collaboration each node is extremely unreliable.

Each node measures its environment using a digital sensor of $m$ levels dynamical range (i.e. $S_n\in\{0,1,\dots,m-1\}$). Under the absence of an attack, the measured signal is assumed to be normally distributed with a particular mean value and variance. For simplicity of the analysis, we assume that $S_n$ conditioned in $\{W=0\}$ distributes following a binomial distribution of parameters $(m,q)$, i.e.
\begin{equation}\label{eq:sadg}
\Pw{0}{S_n=s_n} = \binom{m}{s_n} q^{s_n} (1-q)^{m-s_n} := f(s_n;m,q)
\end{equation}
which because of the central limit theorem approximates a Gaussian variable when $m$ is relatively large. Moreover, it is assumed that the sensor dynamical range is adapted to match the mean value on the lower third of the sensor dynamical range, i.e. $\E{ S_n |W=0} = m/3$. This naturally imposes the requirement $q=1/3$.

Following standard statistical approaches, it is further assumed that the sensors control the environment looking for events where the measured data is significantly high, i.e. when it is larger than mean value in more than two standard deviations. This corresponds e.g. when a specific chemical compound trespasses safe concentration values or when too much movement has been detected over a given time window~(see e.g. \cite{mckenna2008detecting}). Using the fact that $\text{Var}\{S_n\} = mq(1-q)$, this gives a threshold $T = \E{S_n} + 2 \sqrt{ \text{Var}\{S_n\} } = np + 2\sqrt{nq(1-q)}$. Therefore, it is assumed that an attack is related to the event of $S_n$ being uniformly distributed in $[T,m]$. Therefore, one finds that
\begin{align}
\Pw{1}{S_n=s_n} =& 
(1-r) \Pw{1}{ S_n=s_n | \text{\footnotesize{attack out of range}} }
+ r \Pw{1}{ S_n=s_n | \text{\footnotesize{attack in range}} } \nonumber \\
=&  (1-r) f(s_n;m,q) + r \frac{H( s_n - T)}{m-T}, \label{eq:sadg1}
\end{align}
where $H(x)$ is the discrete Heaviside (step) function given by
\begin{equation}
H(x) = 
\begin{cases}
1 \qquad &\text{if } x\geq 0 \\
0 \qquad &\text{in other case.}
\end{cases}
\end{equation}
In summary, $S_n$ conditioned on $\{W=1\}$ is modeled as a mixture model between a Binomial and a truncated uniform distribution, where the relative weight between them is determined by $r$ (c.f. Figure~\ref{fig:signal_dist}). Finally, using \eqref{eq:sadg}  and \eqref{eq:sadg}, the log-likelihood function of the signal $S_n$ can be determined as (see Figure~\ref{fig:loglikelihood})
\begin{equation}
\Lambda_{S_n}(s_n) = \log \frac{ \Pw{1}{S_n = s_n} }{ \Pw{0}{S_n=s_n} }
= \log \left\{ (1-r) + \frac{ r H(s_n - T) }{ (m-T) f(s_n;m,q) } \right \} .
\end{equation} 
%
%

\begin{figure}
  \centering
\begin{tikzpicture}[scale=0.8]
  \begin{axis}
[
 legend pos=north east,
 axis lines*=middle,
 cycle list name = nano mark style,
 xlabel={Digital sensor levels},
 ylabel={Probability},
 grid=both,
 minor grid style={gray!25},
 major grid style={gray!25},
 scaled ticks=false, 
 tick label style={/pgf/number format/fixed}
] 
\addplot table[x=x, y = W0, col sep=comma]{./signal_stuff.csv} ; 
\addlegendentry{$\Pw{0}{S_n=s}$} ; 
\addplot table[x=x, y = W1, col sep=comma]{./signal_stuff.csv} ; 
\addlegendentry{$\Pw{1}{S_n=s}$} ; 
  \end{axis}
\end{tikzpicture}
\caption{Probability distribution for a digital sensor of $m=16$ levels, conditioned on the events $\{W=0\}$ and $\{W=1\}$.}
\label{fig:signal_dist}
\end{figure}
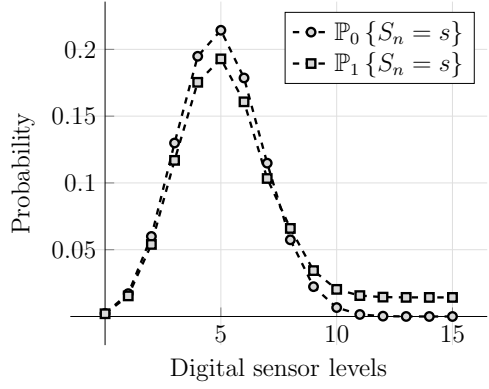

\begin{figure}
\centering
\begin{tikzpicture}[scale=0.8]
  \begin{axis}
[
 cycle list name = nano mark style,
 legend pos=north west,
 axis lines*=middle,
 xlabel={Digital sensor level},
 ylabel={Likelihood},
 grid=both,
 minor grid style={gray!25},
 major grid style={gray!25},
] 
\addplot table[x=x, y = Log-likelihood, col sep=comma]{./signal_stuff.csv} ; 
\addlegendentry{$\Lambda_{S_n}(s)$} ; 
  \end{axis}
\end{tikzpicture}
\caption{Log-likelihood of a digital signal of $m=16$ levels with respect to the variable $W$.}
\label{fig:loglikelihood}
\end{figure}
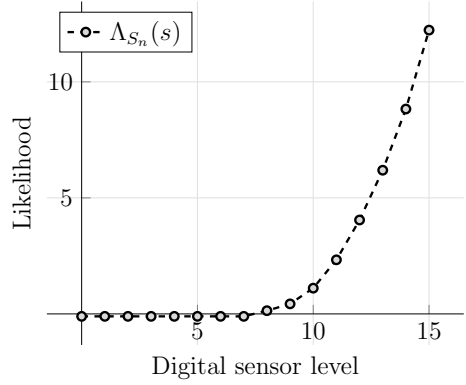

We focus on scenarios where a node can overhear the tranmissions of all the other nodes. However, we consider the case where the listening period is restricted\footnote{It is well-known that the wireless radios of small sensor nodes consume a similar amount of energy while transmitting or receiving data, and hence reducing overhearing periods is key for attaining energy efficiency and hence long network lifetime~\cite{rosas2012modulation}.}. We therefore study the case where the social information gathered by the $n$-th node is $\bs{G}_n = (X_{n-k-1},\dots,X_{n-1})$ if $n > k$. Here $k$ is a design parameter, whose impact in the network performance is studied in the next section.

\subsection{Discussion}
\label{sec:discussion}

We analysed the performance of networks of $N=300$ sensor nodes whose sensors can monitor $r=5\%$ of the target area. Using the definition given in \eqref{eq:MD} and \eqref{eq:FA}, combined with \eqref{eq:tauuuw}, miss-detection and false alarm rates were computed as
\begin{align}
\P{\text{MD}} =& \sum_{\bs{g}\in\mathcal{G}_n} F_1^\Lambda( \tau_n(\bs{g})) \Pw{1}{\bs{G}_n = \bs{g}} \quad \text{and}\\
\P{\text{FA}} =& \sum_{\bs{g}\in\mathcal{G}_n} (1-F_0^\Lambda( \tau_n(\bs{g})) ) \Pw{0}{\bs{G}_n = \bs{g}},
\end{align}
where the terms $\Pw{w}{\bs{G}_n=\bs{g}}$ are computed using Algorithm~1 (c.f. Section~\ref{sec':alg}. In order to favour the reduction of miss-detections over false alarms, $\tau_0=0$ is chosen as is the lowest value that still allows a non-trivial inference process\footnote{Simulations showed that if $\tau<0$ then $X_n=1$ for all $n\in\mathbb{N}$ independently of the value of $W$, triggering a premature information cascade.}. We consider an upper bound over the tolerable false alarm rate of $5\%$.

Simulations demonstrate that the proposed scheme enables strong network resilience in this scenario, allowing the sensor network to maintain a low miss-detection rate even in the presence of an important number of Byzantine nodes (see Figure~\ref{fig:performance}). Please recall that if a traditional distributed detection  scheme based on centralised decision is used, a topology-aware attacker can cause a miss-detection rate of $100\%$ by just compromising the few nodes that perform data aggregation (i.e. the FC(s)). Figure~\ref{fig:performance} shows that nodes that individually would have a miss-detection rate of $95\%$ can improve up to around $10\%$ even when $30\%$ of the nodes are under the control of the attacker. Therefore, by making all the nodes to agregate data, the network can overcome the influence of Bizantine nodes and hence even when some nodes have been compromised the rest of the network can and generate a correct inference.
\begin{figure}[h]
  \centering
\begin{tikzpicture}[scale=0.8]
%
  \begin{axis}
[
 ymode = log,
 mark size=0.7pt,
 legend pos=south west,
 cycle list name = nano line style,
 xlabel={Agent},
 xmax = 200,
 xmin =0,
 ylabel={Rate},
 ymin =0.0001,
 ymax = 1,
 grid=both,
 minor grid style={gray!25},
 major grid style={gray!25},
 width=0.7\linewidth,
 height=0.55\linewidth,
no marks
] 
\addplot table[x=agent, y = MD_0, col sep=comma]{./performance.csv} ; 
\addlegendentry{$N^*/N = 0$} ; 
\addplot table[x=agent, y = MD_0.1, col sep=comma]{./performance.csv} ; 
\addlegendentry{$N^*/N = 0.1$} ; 
\addplot table[x=agent, y = MD_0.3 , col sep=comma]{./performance.csv} ;
\addlegendentry{$N^*/N = 0.3$} ;
\addplot table[x=agent, y = MD_0.5 , col sep=comma]{./performance.csv} ;
\addlegendentry{$N^*/N = 0.5$} ;
%
%
  \end{axis}

\end{tikzpicture}
\caption{Performance for the inference of each node for various attack intensities, given by the average ratio of Byzantine nodes $N^*/N = p_b$. Agents overhear the previous $k=4$ broadcasted signals, and use sensors with dynamical range of $n=64$.}
\label{fig:performance}
\end{figure}
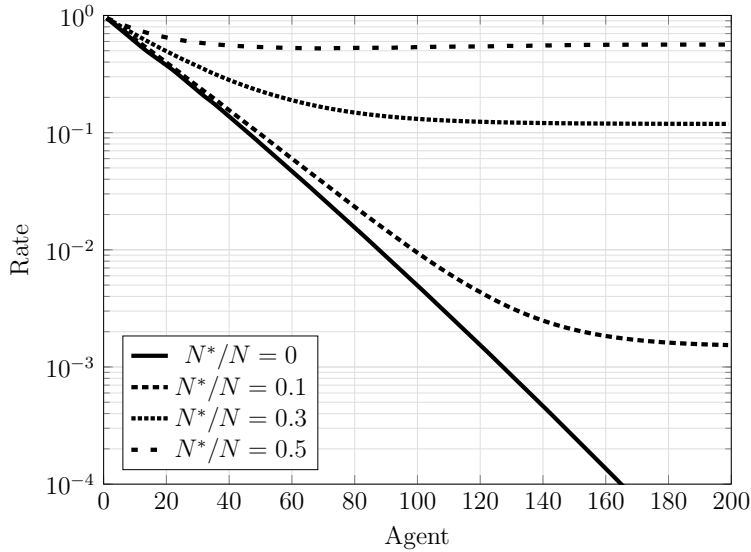

Please note that, for the case illustrated by  Figure~\ref{fig:performance}, when there are Byzantine nodes the miss-detection rate improves until the network size reaches $N=500$, achieving a performance of $\approx 10^{-12}$ (not shown in the Figure). This result has two important implications. First, this confirms the prediction of Theorem~\ref{teo2} that if the signal log-likelihood is bounded then information cascades are eventually dominant, hence stoping the learning process of the network (for a more detailed discussion about this issue please c.f. \cite{rosas2017technological}). Secondly, this result stress a key difference of our approach with respect to the existent literature about information cascades: \emph{even if information cascades become dominant and hence perfect social learning cannot be achieved by bounded signals, the achieved performance can still be very high and hence useful in a practical information-processing setup}.  

%


The network resilience provided by our scheme is influenced by the sensor dynamical range, $m$, as a higher sensor resolution is likely to provide more discriminative power. Our results show three sharply distinct regimes (see Figure~\ref{fig:signal_range}). First, if $m$ is too small ($m\leq 4$) the network performance is very poor, irrespective of the number of Byzantine nodes. Secondly, if $8\leq m \leq 32$ the miss-detection rate without Byzantine nodes is approx. $10\%$ (cf. Figure~\ref{fig:signal_range}) and is exponentially degraded by the presence of Byzantine nodes. Finally, if $m\geq 64$ then the performance under no Byzantine nodes is very high, and is degraded super-exponentially by the presence of Byzantine nodes. Interestingly, the point at which the miss-detection rate of this regime goes above $10^{-1}$ is at $N^*/N=1/3$, having some resemblance with the well-known $1/3$ threshold of the Byzantine generals problem~\cite{lamport1982byzantine}. Also, it is intriguing the fact that differences between $8$ and $32$ levels in the dynamical range gives practically no performance benefits.
\begin{figure}[h]
  \centering
\begin{tikzpicture}[scale=0.8, mark options={solid,scale=2,fill=gray!30}]
  \begin{axis}
[
 ymode = log,
 legend pos=south east,
 xlabel={Attack intensity ($N^*/N$)},
 xmax = 1,
 xmin =0,
 ylabel={Miss-detection rate},
 ymin =0.00000001,
 ymax = 1,
 grid=both,
 minor grid style={gray!25},
 major grid style={gray!25},
 width=0.7\linewidth,
 height=0.55\linewidth,
cycle list name = nano mark style
] 
\addplot table[x=x, y =Range 2, col sep=comma]{./range2.csv} ; 
\addlegendentry{$m=2$} ; 
\addplot table[x=x, y =Range 16, col sep=comma]{./range16.csv} ; 
\addlegendentry{$m=16$} ; 
\addplot table[x=x, y =Range 32, col sep=comma]{./range32.csv} ; 
\addlegendentry{$m=32$} ; 
\addplot table[x=x, y =Range 64, col sep=comma]{./range64.csv} ; 
\addlegendentry{$m=64$} 
\addplot table[x=x, y =Range 256, col sep=comma]{./range256.csv} ; 
\addlegendentry{$m=256$} 
\end{axis}
\end{tikzpicture}
\caption{Effect of the sensor dynamical range over the network resilience.}
\label{fig:signal_range}
\end{figure}
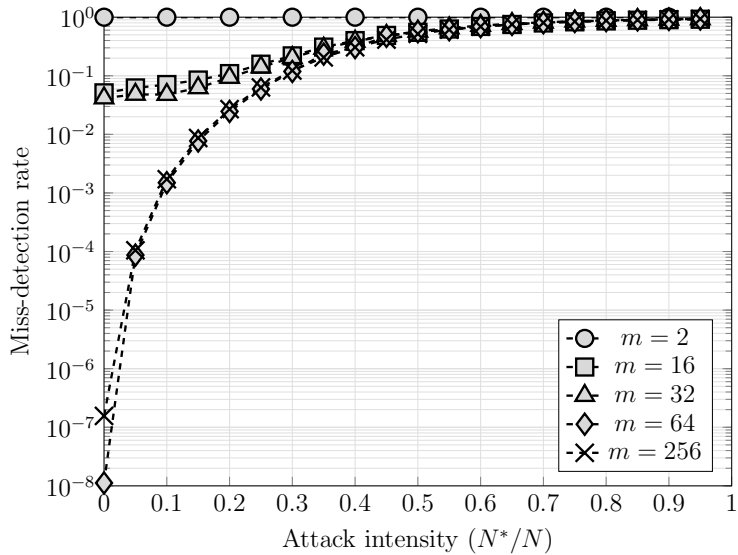

Our results also show the effects of the memory size, $k$, showing that larger values of $k$ provides great benefits for the network resilience (see Figure~\ref{fig:memory}). In effect, by performing an optimal bayesian inference over 8 broadcasted signals the network miss-detection rate remains bellow $10\%$ up to an attack intensity of $50\%$ of Byzantine nodes. Unfortunately, the computation and storage requirements of Algorithm~1 grow exponentially with $k$, and hence using memories beyond $k=10$ is not practical for resource-limited sensor networks. Overcoming this limitation is an relevant future line of investigation.
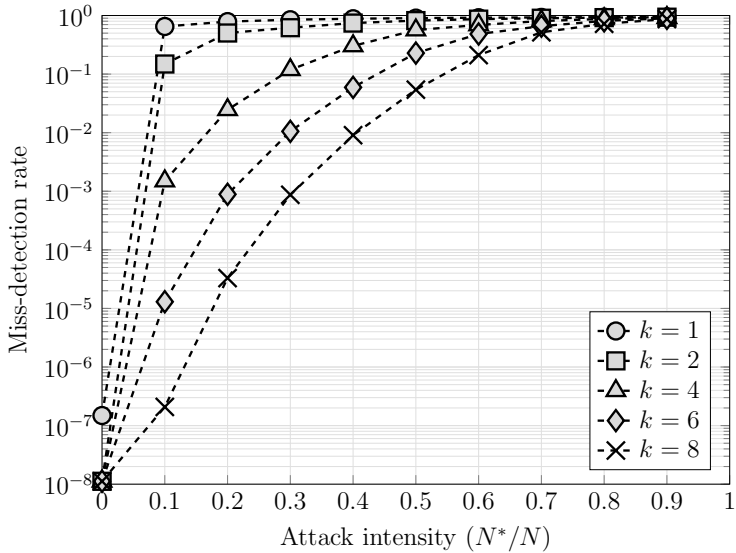
\begin{figure}[h]
  \centering
\begin{tikzpicture}[scale=0.8, mark options={solid,scale=2,fill=gray!30}]
  \begin{axis}
[
 ymode = log,
 legend pos=south east,
 xlabel={Attack intensity ($N^*/N$)},
 xmax = 1,
 xmin =0,
 ylabel={Miss-detection rate},
 ymin =0.00000001,
 ymax = 1,
 grid=both,
 minor grid style={gray!25},
 major grid style={gray!25},
 width=0.7\linewidth,
 height=0.55\linewidth,
cycle list name = nano mark style
] 
\addplot table[x=x, y =Memory 1, col sep=comma]{./memory1.csv} ; 
\addlegendentry{$k=1$} ; 
%
%
\addplot table[x=x, y =Memory 2, col sep=comma]{./memory2.csv} ; 
\addlegendentry{$k=2$} ; 
\addplot table[x=x, y =Memory 4, col sep=comma]{./memory4.csv} ; 
\addlegendentry{$k=4$} ; 
\addplot table[x=x, y =Memory 6, col sep=comma]{./memory6.csv} ; 
\addlegendentry{$k=6$} 
\addplot table[x=x, y =Memory 8, col sep=comma]{./memory8.csv} ; 
\addlegendentry{$k=8$} 
\end{axis}
\end{tikzpicture}
\caption{A larger node memory, which allows to include more social signals into the inference process, greatly improves the network resilience.}
\label{fig:memory}
\end{figure}

\section{Conclusions}
\label{sec:5}

Traditional approaches to data aggregation over information networks are based on a strong division of labour, which discriminates between sensing nodes that merely sense and forward data, and fusion centers that monopolize all the processing and inference capabilities. This generates a single point of failure, which is likely to be exploited by smart adversaries whose interest is the disruption of the network capabilities.

This serious security thread can be overcome by distributing the decision making process across the network using social learning principles. This approach avoids single points of failure by generating a large number of nodes from where aggregated data can be accessed. In this paper a social learning data fusion scheme has been proposed, which is suitable of being implemented in devices with limited computational capabilities.

We showed that if the private signals are bounded then each local information cascades triggers a global cascade, extending previous results to the case where an adversary controles an number of Byzantine nodes. This result is highly relevant for sensor networks, as digital sensors are intrinsically bounded and hence satisfy the assumptions of these results. However, contrasting with the literature, our approach does not focus on the conditions that guarantee perfect asymptotical social learning (i.e. miss-detection and false alarm rates converging to zero), but if their limit is small enough for practical applications. Our results show that this is the case, even when limiting the number of overheared transmissions.

Moreover, our results suggest that social learning principles can enable significant resilience of an information network against topology-aware data falsification attacks, which can totally disable the detection capabilities of traditional sensor networks. Furthermore, our results illustrate how the network resilience can persists even when the attacker has compromised an important number of nodes.

It is our hope that these results can motivate further explorations on the interface between distributed decision making, statistical inference and signal processing over technological and social networks.

\appendix

\section{Properties of $F_w^\Lambda$}
\label{sec:appendix_sig}

For simplicity let us consider the case of real-value signals, i.e. $S_n\in \mathbb{R}$. In this case, the c.d.f. of the signal likelihood is given by
\begin{equation}\label{eq:asdqwrgty}
F_w^\Lambda(y) = \int_{\mathcal{S}^y} \text{d}\mu_w
\end{equation}
where $\mathcal{S}^y = \{ x\in \mathbb{R} | \Lambda_s(x) \leq y \}$. If $\Lambda_s$ is an increasing function, then $\mathcal{S}^y=\{x\in\mathbb{R} |x \leq \Lambda_s^{-1}(y) \} = (-\infty, \Lambda_s^{-1}(y) ]$ and hence
\begin{equation}
F_w^\Lambda(y) = \int_{-\infty}^{\Lambda_s^{-1}(y)} \text{d}\mu_w = H_w( \Lambda_s^{-1}(y) )
\enspace,
\end{equation}
where $H_w(s) $ is the cumulative density function (c.d.f.) of $S_n$ for $W=w$. For the general case where $\Lambda_s$ is an arbitrary (piece-wise continuous) function, then $\mathcal{S}^y$ can be expressed as the union of intervals. Then $\cup_{j=1}^{\infty} [a_j(y),b_j(y)] = \mathcal{S}^y$ (note that $\Lambda_s(a_j(y))=\Lambda_s(b_k(y))=y$) and hence from \eqref{eq:asdqwrgty} is clear that
\begin{equation}\label{eq:asdqwety}
F_w^\Lambda(y) = \sum_{j=1}^\infty \int_{a_j(y)}^{b_j(y)} \text{d}\mu_w = \sum_{j=1}^\infty \left[ H_w( b_j(y)) - H_w(a_j(y)) \right]
\enspace.
\end{equation}

\section{Proof of Lemma~\ref{lemma12123123}}
\label{app:2}

\begin{proof} Lets assume that the process $\bs{G}_n$ has strong consistent transitions and consider $\bs{g'}\in\mathcal{G}_{n-1}$ such that $\tau_{n-1}(\bs{g'}) \leq L_s$. Note that, under these conditions $F_w^\Lambda(\tau_{n-1}(\bs{g'})) = 0$, and hence
\begin{equation}
  \Pw{w}{X_{n-1}=1|\bs{G}_{n-1}=\bs{g'}} = 1 - z_0 - z_1 F_w^\Lambda(\tau_{n-1}(\bs{g'})) = 1 - p_b c_{0|1} = 0
\end{equation}
holds for any $w\in\{0,1\}$. Moreover, this allows to find that
\begin{align}
\mathbb{P}_w\{ \bs{G}_n=\bs{g}|\bs{G}_{n-1}=\bs{g'}\} 
&= \sum_{x_n\in\{0,1\}} \beta_w^n(\bs{g}|x_n,\bs{g'}) \Pw{w}{X_{n-1}=x_n|\bs{G}_{n-1}=\bs{g'}} \nonumber\\
&= \beta_w^n(\bs{g}|1,\bs{g'})
\enspace.
\end{align}
Therefore, due to the strongly consistent transition property, if $\Pw{w}{\bs{G}_n=\bs{g}|\bs{G}_{n-1}=\bs{g'}} = \beta_w^n(\bs{g}|1,\bs{g'}) > 0$ then
\begin{equation}
L_s \geq \tau_{n-1}(\bs{g'}) \geq \tau_n(\bs{g})
\enspace,
\end{equation}
proving the weak consistent transition property. The proof for the case of $\tau_{n-1}(\bs{g'}) \geq U_s$ is analogous.
\end{proof}

\section{Proof of Lemma~\ref{ffewafawef}}
\label{app:3}

\begin{proof} 
Let us consider $\bs{g}_0\in\mathcal{G}_{n}$ such that $\tau_{n}(\bs{g}_0) \notin [L_s,U_s]$. Then, due to the weakly invertible evolution, for each $x\in\{0,1\}$ there exists $\bs{g}(x)\in\mathcal{G}_{n+1}$ such that
\begin{equation}
\beta_w^n(\bs{g}|x,\bs{g}_0) = 
\begin{cases}
1 \qquad &\text{if } \bs{g}=\bs{g}(x),\\
0 \qquad &\text{in other case.}
\end{cases}
\end{equation}
Moreover, note that while the deterministic assumption implies that the event $\{\bs{G}_{n}=\bs{g}_0\}$ could be followed by either $\{\bs{G}_{n+1}=\bs{g}(0)\}$ or $\{\bs{G}_{n+1}=\bs{g}(1)\}$, the 1-1 assumption requires that $\bs{g}(0) = \bs{g}(1)$. With this, note that
\begin{align}
\Lambda_{\bs{G}_{n+1}}(\bs{g}(0)) =& \log \frac{ \Pw{1}{\bs{G}_{n+1} = \bs{g}(0) }}{ \Pw{0}{\bs{G}_{n+1} = \bs{g}(0) }}  \nonumber\\
=& \log \frac{ \sum_{\substack{\bs{g'}\in\mathcal{G}_{n}\\ x \in\{0,1\}}} \beta_w^n(\bs{g}(x)|x,\bs{g'}) \Pw{1}{X_n=x,\bs{G}_{n}=\bs{g'}} }{ \sum_{\substack{\bs{g'}\in\mathcal{G}_{n}\\ x\in\{0,1\}}} \beta_w^n(\bs{g}(x)|x,\bs{g'}) \Pw{0}{X_n=x,\bs{G}_{n}=\bs{g'}} } \label{eqerr} \\
=& \log \frac{ \sum_{x\in\{0,1\}} \Pw{1}{X_n=x|\bs{G}_n=\bs{g}_0} \Pw{1}{ \bs{G}_n = \bs{g}_0} }{ \sum_{x\in\{0,1\}} \Pw{0}{X_n=x|\bs{G}_n=\bs{g}_0} \Pw{0}{ \bs{G}_n = \bs{g}_0} } \label{eqerr1} \\
=& \Lambda_{\bs{G}_{n-1}}(\bs{g}_0), \label{eqerr2}
\end{align}
Above, \eqref{eqerr} is a consequence of $\bs{g}(0) = \bs{g}(1)$, while \eqref{eqerr1} is because of the 1-1 condition over the dynamic. Finally, to justify \eqref{eqerr2} let us first consider
\begin{equation}
\Pw{w}{X_n=x|\bs{G}_n=\bs{g}_0} = \lambda( z_0 +  z_1 F_w^\Lambda(\tau_n(\bs{g}_0)),x) .
\end{equation}
Because $\tau_n(\bs{g}_0) \notin [L_s,U_s]$ then $F_w^\Lambda(\tau_n(\bs{g}_0))$ is either 0 or 1; in any case it does not depends on $W$. This, in turn means that $\Pw{1}{X_n=x|\bs{G}_n=\bs{g}_0} = \Pw{0}{X_n=x|\bs{G}_n=\bs{g}_0}$, which explains how \eqref{eqerr2} is obtained.

Please note that \eqref{eqerr2} shows that, once $\tau_n$ leaves $[L_s,U_s]$, it keeps a constant value. This, in turn, shows that weakly deterministic transitions satisfy the weakly consistency condition.
\end{proof}

\section*{Acknowledgements}

Fernando Rosas is supported by the European Union's H2020 research and innovation programme, under the Marie Sk\l{}odowska-Curie grant agreement No. 702981.

\newcommand{\BMCxmlcomment}[1]{}

\BMCxmlcomment{

<refgrp>

<bibl id="B1">
  <title><p>Cyber--physical systems: A perspective at the
  centennial</p></title>
  <aug>
    <au><snm>Kim</snm><fnm>KD</fnm></au>
    <au><snm>Kumar</snm><fnm>PR</fnm></au>
  </aug>
  <source>Proceedings of the IEEE</source>
  <publisher>IEEE</publisher>
  <pubdate>2012</pubdate>
  <volume>100</volume>
  <issue>Special Centennial Issue</issue>
  <fpage>1287</fpage>
  <lpage>-1308</lpage>
</bibl>

<bibl id="B2">
  <title><p>What you need to know about the WannaCry Ransomware</p></title>
  <aug>
    <au><snm>Response</snm><fnm>SS</fnm></au>
  </aug>
</bibl>

<bibl id="B3">
  <title><p>Distributed inference in wireless sensor networks</p></title>
  <aug>
    <au><snm>Veeravalli</snm><fnm>VV</fnm></au>
    <au><snm>Varshney</snm><fnm>PK</fnm></au>
  </aug>
  <source>Philosophical Transactions of the Royal Society of London A:
  Mathematical, Physical and Engineering Sciences</source>
  <publisher>The Royal Society</publisher>
  <pubdate>2012</pubdate>
  <volume>370</volume>
  <issue>1958</issue>
  <fpage>100</fpage>
  <lpage>-117</lpage>
</bibl>

<bibl id="B4">
  <title><p>Distributed Detection and Estimation in Wireless</p></title>
  <aug>
    <au><snm>Barbarossa</snm><fnm>S</fnm></au>
    <au><snm>Sardellitti</snm><fnm>S</fnm></au>
    <au><snm>Di Lorenzo</snm><fnm>P</fnm></au>
  </aug>
  <source>Academic Press Library in Signal Processing: Communications and Radar
  Signal Processing</source>
  <publisher>Academic Press</publisher>
  <pubdate>2013</pubdate>
  <volume>2</volume>
  <fpage>329</fpage>
</bibl>

<bibl id="B5">
  <title><p>The role of advanced sensing in smart cities</p></title>
  <aug>
    <au><snm>Hancke</snm><fnm>GP</fnm></au>
    <au><snm>Hancke Jr</snm><fnm>GP</fnm></au>
    <au><cnm>others</cnm></au>
  </aug>
  <source>Sensors</source>
  <publisher>Multidisciplinary Digital Publishing Institute</publisher>
  <pubdate>2012</pubdate>
  <volume>13</volume>
  <issue>1</issue>
  <fpage>393</fpage>
  <lpage>-425</lpage>
</bibl>

<bibl id="B6">
  <title><p>Scalable anomaly detection for smart city infrastructure
  networks</p></title>
  <aug>
    <au><snm>Difallah</snm><fnm>DE</fnm></au>
    <au><snm>Cudre Mauroux</snm><fnm>P</fnm></au>
    <au><snm>McKenna</snm><fnm>SA</fnm></au>
  </aug>
  <source>IEEE Internet Computing</source>
  <publisher>IEEE</publisher>
  <pubdate>2013</pubdate>
  <volume>17</volume>
  <issue>6</issue>
  <fpage>39</fpage>
  <lpage>-47</lpage>
</bibl>

<bibl id="B7">
  <title><p>Contamination detection in drinking water distribution systems
  using sensor networks</p></title>
  <aug>
    <au><snm>Lambrou</snm><fnm>TP</fnm></au>
    <au><snm>Panayiotou</snm><fnm>CG</fnm></au>
    <au><snm>Polycarpou</snm><fnm>MM</fnm></au>
  </aug>
  <source>Control Conference (ECC), 2015 European</source>
  <pubdate>2015</pubdate>
  <fpage>3298</fpage>
  <lpage>-3303</lpage>
</bibl>

<bibl id="B8">
  <title><p>A low-cost sensor network for real-time monitoring and
  contamination detection in drinking water distribution systems</p></title>
  <aug>
    <au><snm>Lambrou</snm><fnm>TP</fnm></au>
    <au><snm>Anastasiou</snm><fnm>CC</fnm></au>
    <au><snm>Panayiotou</snm><fnm>CG</fnm></au>
    <au><snm>Polycarpou</snm><fnm>MM</fnm></au>
  </aug>
  <source>IEEE sensors journal</source>
  <publisher>IEEE</publisher>
  <pubdate>2014</pubdate>
  <volume>14</volume>
  <issue>8</issue>
  <fpage>2765</fpage>
  <lpage>-2772</lpage>
</bibl>

<bibl id="B9">
  <title><p>Security in wireless sensor networks</p></title>
  <aug>
    <au><snm>Perrig</snm><fnm>A</fnm></au>
    <au><snm>Stankovic</snm><fnm>J</fnm></au>
    <au><snm>Wagner</snm><fnm>D</fnm></au>
  </aug>
  <source>Communications of the ACM</source>
  <publisher>ACM</publisher>
  <pubdate>2004</pubdate>
  <volume>47</volume>
  <issue>6</issue>
  <fpage>53</fpage>
  <lpage>-57</lpage>
</bibl>

<bibl id="B10">
  <title><p>Designing secure sensor networks</p></title>
  <aug>
    <au><snm>Shi</snm><fnm>E</fnm></au>
    <au><snm>Perrig</snm><fnm>A</fnm></au>
  </aug>
  <source>IEEE Wireless Communications</source>
  <publisher>IEEE</publisher>
  <pubdate>2004</pubdate>
  <volume>11</volume>
  <issue>6</issue>
  <fpage>38</fpage>
  <lpage>-43</lpage>
</bibl>

<bibl id="B11">
  <title><p>Security in wireless sensor networks: issues and
  challenges</p></title>
  <aug>
    <au><snm>Pathan</snm><fnm>ASK</fnm></au>
    <au><snm>Lee</snm><fnm>HW</fnm></au>
    <au><snm>Hong</snm><fnm>CS</fnm></au>
  </aug>
  <source>Advanced Communication Technology, 2006. ICACT 2006. The 8th
  International Conference</source>
  <pubdate>2006</pubdate>
  <volume>2</volume>
  <fpage>6</fpage>
  <lpage>-pp</lpage>
</bibl>

<bibl id="B12">
  <title><p>Low-Energy Security: Limits and Opportunities in the Internet of
  Things</p></title>
  <aug>
    <au><snm>Trappe</snm><fnm>W.</fnm></au>
    <au><snm>Howard</snm><fnm>R.</fnm></au>
    <au><snm>Moore</snm><fnm>R. S.</fnm></au>
  </aug>
  <source>IEEE Security Privacy</source>
  <pubdate>2015</pubdate>
  <volume>13</volume>
  <issue>1</issue>
  <fpage>14</fpage>
  <lpage>21</lpage>
</bibl>

<bibl id="B13">
  <title><p>Distributed detection in the presence of byzantine
  attacks</p></title>
  <aug>
    <au><snm>Marano</snm><fnm>S</fnm></au>
    <au><snm>Matta</snm><fnm>V</fnm></au>
    <au><snm>Tong</snm><fnm>L</fnm></au>
  </aug>
  <source>IEEE Transactions on Signal Processing</source>
  <publisher>IEEE</publisher>
  <pubdate>2009</pubdate>
  <volume>57</volume>
  <issue>1</issue>
  <fpage>16</fpage>
  <lpage>-29</lpage>
</bibl>

<bibl id="B14">
  <title><p>The Byzantine generals problem</p></title>
  <aug>
    <au><snm>Lamport</snm><fnm>L</fnm></au>
    <au><snm>Shostak</snm><fnm>R</fnm></au>
    <au><snm>Pease</snm><fnm>M</fnm></au>
  </aug>
  <source>ACM Transactions on Programming Languages and Systems
  (TOPLAS)</source>
  <publisher>ACM</publisher>
  <pubdate>1982</pubdate>
  <volume>4</volume>
  <issue>3</issue>
  <fpage>382</fpage>
  <lpage>-401</lpage>
</bibl>

<bibl id="B15">
  <title><p>Distributed inference with byzantine data: State-of-the-art review
  on data falsification attacks</p></title>
  <aug>
    <au><snm>Vempaty</snm><fnm>A</fnm></au>
    <au><snm>Tong</snm><fnm>L</fnm></au>
    <au><snm>Varshney</snm><fnm>PK</fnm></au>
  </aug>
  <source>IEEE Signal Processing Magazine</source>
  <publisher>IEEE</publisher>
  <pubdate>2013</pubdate>
  <volume>30</volume>
  <issue>5</issue>
  <fpage>65</fpage>
  <lpage>-75</lpage>
</bibl>

<bibl id="B16">
  <title><p>Distributed Inference With M-Ary Quantized Data in the Presence of
  Byzantine Attacks</p></title>
  <aug>
    <au><snm>Nadendla</snm><fnm>V. S. S.</fnm></au>
    <au><snm>Han</snm><fnm>Y. S.</fnm></au>
    <au><snm>Varshney</snm><fnm>P. K.</fnm></au>
  </aug>
  <source>IEEE Transactions on Signal Processing</source>
  <pubdate>2014</pubdate>
  <volume>62</volume>
  <issue>10</issue>
  <fpage>2681</fpage>
  <lpage>2695</lpage>
</bibl>

<bibl id="B17">
  <title><p>Asymptotically Optimum Distributed Estimation in the Presence of
  Attacks</p></title>
  <aug>
    <au><snm>Zhang</snm><fnm>J.</fnm></au>
    <au><snm>Blum</snm><fnm>R. S.</fnm></au>
    <au><snm>Lu</snm><fnm>X.</fnm></au>
    <au><snm>Conus</snm><fnm>D.</fnm></au>
  </aug>
  <source>IEEE Transactions on Signal Processing</source>
  <pubdate>2015</pubdate>
  <volume>63</volume>
  <issue>5</issue>
  <fpage>1086</fpage>
  <lpage>1101</lpage>
</bibl>

<bibl id="B18">
  <title><p>Distributed Bayesian Detection in the Presence of Byzantine
  Data</p></title>
  <aug>
    <au><snm>Kailkhura</snm><fnm>B.</fnm></au>
    <au><snm>Han</snm><fnm>Y. S.</fnm></au>
    <au><snm>Brahma</snm><fnm>S.</fnm></au>
    <au><snm>Varshney</snm><fnm>P. K.</fnm></au>
  </aug>
  <source>IEEE Transactions on Signal Processing</source>
  <pubdate>2015</pubdate>
  <volume>63</volume>
  <issue>19</issue>
  <fpage>5250</fpage>
  <lpage>5263</lpage>
</bibl>

<bibl id="B19">
  <title><p>Distributed detection in tree topologies with
  byzantines</p></title>
  <aug>
    <au><snm>Kailkhura</snm><fnm>B</fnm></au>
    <au><snm>Brahma</snm><fnm>S</fnm></au>
    <au><snm>Han</snm><fnm>YS</fnm></au>
    <au><snm>Varshney</snm><fnm>PK</fnm></au>
  </aug>
  <source>IEEE Transactions on Signal Processing</source>
  <publisher>IEEE</publisher>
  <pubdate>2014</pubdate>
  <volume>62</volume>
  <issue>12</issue>
  <fpage>3208</fpage>
  <lpage>-3219</lpage>
</bibl>

<bibl id="B20">
  <title><p>Distributed Detection in Tree Networks: Byzantines and Mitigation
  Techniques</p></title>
  <aug>
    <au><snm>Kailkhura</snm><fnm>B.</fnm></au>
    <au><snm>Brahma</snm><fnm>S.</fnm></au>
    <au><snm>Dulek</snm><fnm>B.</fnm></au>
    <au><snm>Han</snm><fnm>Y. S.</fnm></au>
    <au><snm>Varshney</snm><fnm>P. K.</fnm></au>
  </aug>
  <source>IEEE Transactions on Information Forensics and Security</source>
  <pubdate>2015</pubdate>
  <volume>10</volume>
  <issue>7</issue>
  <fpage>1499</fpage>
  <lpage>1512</lpage>
</bibl>

<bibl id="B21">
  <title><p>Machine-to-machine communications: Technologies and
  challenges</p></title>
  <aug>
    <au><snm>Chen</snm><fnm>KC</fnm></au>
    <au><snm>Lien</snm><fnm>SY</fnm></au>
  </aug>
  <source>Ad Hoc Networks</source>
  <publisher>Elsevier</publisher>
  <pubdate>2014</pubdate>
  <volume>18</volume>
  <fpage>3</fpage>
  <lpage>-23</lpage>
</bibl>

<bibl id="B22">
  <title><p>Distributed detection of node replication attacks in sensor
  networks</p></title>
  <aug>
    <au><snm>Parno</snm><fnm>B</fnm></au>
    <au><snm>Perrig</snm><fnm>A</fnm></au>
    <au><snm>Gligor</snm><fnm>V</fnm></au>
  </aug>
  <source>2005 IEEE Symposium on Security and Privacy (S\&P'05)</source>
  <pubdate>2005</pubdate>
  <fpage>49</fpage>
  <lpage>-63</lpage>
</bibl>

<bibl id="B23">
  <title><p>Improving spectrum efficiency via in-network computations in
  cognitive radio sensor networks</p></title>
  <aug>
    <au><snm>Lin</snm><fnm>SC</fnm></au>
    <au><snm>Chen</snm><fnm>KC</fnm></au>
  </aug>
  <source>IEEE Transactions on Wireless Communications</source>
  <publisher>IEEE</publisher>
  <pubdate>2014</pubdate>
  <volume>13</volume>
  <issue>3</issue>
  <fpage>1222</fpage>
  <lpage>-1234</lpage>
</bibl>

<bibl id="B24">
  <title><p>Quantifying collectivity</p></title>
  <aug>
    <au><snm>Daniels</snm><fnm>BC</fnm></au>
    <au><snm>Ellison</snm><fnm>CJ</fnm></au>
    <au><snm>Krakauer</snm><fnm>DC</fnm></au>
    <au><snm>Flack</snm><fnm>JC</fnm></au>
  </aug>
  <source>Current opinion in neurobiology</source>
  <publisher>Elsevier</publisher>
  <pubdate>2016</pubdate>
  <volume>37</volume>
  <fpage>106</fpage>
  <lpage>-113</lpage>
</bibl>

<bibl id="B25">
  <title><p>Conflicts of interest improve collective computation of adaptive
  social structures</p></title>
  <aug>
    <au><snm>Brush</snm><fnm>ER</fnm></au>
    <au><snm>Krakauer</snm><fnm>DC</fnm></au>
    <au><snm>Flack</snm><fnm>JC</fnm></au>
  </aug>
  <source>Science Advances</source>
  <publisher>American Association for the Advancement of Science</publisher>
  <pubdate>2018</pubdate>
  <volume>4</volume>
  <issue>1</issue>
  <fpage>e1603311</fpage>
</bibl>

<bibl id="B26">
  <title><p>Decentralized detection</p></title>
  <aug>
    <au><snm>Tsitsiklis</snm><fnm>JN</fnm></au>
    <au><cnm>others</cnm></au>
  </aug>
  <source>Advances in Statistical Signal Processing</source>
  <pubdate>1993</pubdate>
  <volume>2</volume>
  <issue>2</issue>
  <fpage>297</fpage>
  <lpage>-344</lpage>
</bibl>

<bibl id="B27">
  <title><p>Distributed detection with multiple sensors I.
  Fundamentals</p></title>
  <aug>
    <au><snm>Viswanathan</snm><fnm>R</fnm></au>
    <au><snm>Varshney</snm><fnm>PK</fnm></au>
  </aug>
  <source>Proceedings of the IEEE</source>
  <publisher>IEEE</publisher>
  <pubdate>1997</pubdate>
  <volume>85</volume>
  <issue>1</issue>
  <fpage>54</fpage>
  <lpage>-63</lpage>
</bibl>

<bibl id="B28">
  <title><p>Distributed detection with multiple sensors i. advanced
  topics</p></title>
  <aug>
    <au><snm>Blum</snm><fnm>RS</fnm></au>
    <au><snm>Kassam</snm><fnm>SA</fnm></au>
    <au><snm>Poor</snm><fnm>HV</fnm></au>
  </aug>
  <source>Proceedings of the IEEE</source>
  <publisher>IEEE</publisher>
  <pubdate>1997</pubdate>
  <volume>85</volume>
  <issue>1</issue>
  <fpage>64</fpage>
  <lpage>-79</lpage>
</bibl>

<bibl id="B29">
  <title><p>Channel aware distributed detection in wireless sensor
  networks</p></title>
  <aug>
    <au><snm>Chen</snm><fnm>B</fnm></au>
    <au><snm>Tong</snm><fnm>L</fnm></au>
    <au><snm>Varshney</snm><fnm>PK</fnm></au>
  </aug>
  <source>IEEE Signal Processing Mag</source>
  <pubdate>2006</pubdate>
</bibl>

<bibl id="B30">
  <title><p>Wireless sensors in distributed detection applications</p></title>
  <aug>
    <au><snm>Chamberland</snm><fnm>JF</fnm></au>
    <au><snm>Veeravalli</snm><fnm>VV</fnm></au>
  </aug>
  <source>IEEE signal processing magazine</source>
  <publisher>IEEE</publisher>
  <pubdate>2007</pubdate>
  <volume>24</volume>
  <issue>3</issue>
  <fpage>16</fpage>
  <lpage>-25</lpage>
</bibl>

<bibl id="B31">
  <title><p>On the complexity of decentralized decision making and detection
  problems</p></title>
  <aug>
    <au><snm>Tsitsiklis</snm><fnm>J</fnm></au>
    <au><snm>Athans</snm><fnm>M</fnm></au>
  </aug>
  <source>IEEE Transactions on Automatic Control</source>
  <publisher>IEEE</publisher>
  <pubdate>1985</pubdate>
  <volume>30</volume>
  <issue>5</issue>
  <fpage>440</fpage>
  <lpage>-446</lpage>
</bibl>

<bibl id="B32">
  <title><p>Optimum quantization for detector fusion: some proofs, examples,
  and pathology</p></title>
  <aug>
    <au><snm>Warren</snm><fnm>D</fnm></au>
    <au><snm>Willett</snm><fnm>P</fnm></au>
  </aug>
  <source>Journal of the Franklin Institute</source>
  <publisher>Elsevier</publisher>
  <pubdate>1999</pubdate>
  <volume>336</volume>
  <issue>2</issue>
  <fpage>323</fpage>
  <lpage>-359</lpage>
</bibl>

<bibl id="B33">
  <title><p>Asymptotic results for decentralized detection in power constrained
  wireless sensor networks</p></title>
  <aug>
    <au><snm>Chamberland</snm><fnm>J F</fnm></au>
    <au><snm>Veeravalli</snm><fnm>VV</fnm></au>
  </aug>
  <source>IEEE Journal on selected areas in communications</source>
  <publisher>IEEE</publisher>
  <pubdate>2004</pubdate>
  <volume>22</volume>
  <issue>6</issue>
  <fpage>1007</fpage>
  <lpage>-1015</lpage>
</bibl>

<bibl id="B34">
  <title><p>Networks, crowds, and markets</p></title>
  <aug>
    <au><snm>Easley</snm><fnm>D</fnm></au>
    <au><snm>Kleinberg</snm><fnm>J</fnm></au>
  </aug>
  <source>Cambridge University Press</source>
  <publisher>Cambridge Univ Press</publisher>
  <pubdate>2010</pubdate>
  <volume>1</volume>
  <issue>2.1</issue>
  <fpage>2</fpage>
  <lpage>-1</lpage>
</bibl>

<bibl id="B35">
  <title><p>Opinion dynamics and learning in social networks</p></title>
  <aug>
    <au><snm>Acemoglu</snm><fnm>D</fnm></au>
    <au><snm>Ozdaglar</snm><fnm>A</fnm></au>
  </aug>
  <source>Dynamic Games and Applications</source>
  <publisher>Springer</publisher>
  <pubdate>2011</pubdate>
  <volume>1</volume>
  <issue>1</issue>
  <fpage>3</fpage>
  <lpage>-49</lpage>
</bibl>

<bibl id="B36">
  <title><p>A simple model of herd behavior</p></title>
  <aug>
    <au><snm>Banerjee</snm><fnm>AV</fnm></au>
  </aug>
  <source>The Quarterly Journal of Economics</source>
  <publisher>JSTOR</publisher>
  <pubdate>1992</pubdate>
  <fpage>797</fpage>
  <lpage>-817</lpage>
</bibl>

<bibl id="B37">
  <title><p>A theory of fads, fashion, custom, and cultural change as
  informational cascades</p></title>
  <aug>
    <au><snm>Bikhchandani</snm><fnm>S</fnm></au>
    <au><snm>Hirshleifer</snm><fnm>D</fnm></au>
    <au><snm>Welch</snm><fnm>I</fnm></au>
  </aug>
  <source>Journal of political Economy</source>
  <publisher>JSTOR</publisher>
  <pubdate>1992</pubdate>
  <fpage>992</fpage>
  <lpage>-1026</lpage>
</bibl>

<bibl id="B38">
  <title><p>Learning from the behavior of others: Conformity, fads, and
  informational cascades</p></title>
  <aug>
    <au><snm>Bikhchandani</snm><fnm>S</fnm></au>
    <au><snm>Hirshleifer</snm><fnm>D</fnm></au>
    <au><snm>Welch</snm><fnm>I</fnm></au>
  </aug>
  <source>The Journal of Economic Perspectives</source>
  <publisher>JSTOR</publisher>
  <pubdate>1998</pubdate>
  <volume>12</volume>
  <issue>3</issue>
  <fpage>151</fpage>
  <lpage>-170</lpage>
</bibl>

<bibl id="B39">
  <title><p>Pathological outcomes of observational learning</p></title>
  <aug>
    <au><snm>Smith</snm><fnm>L</fnm></au>
    <au><snm>S{\o}rensen</snm><fnm>P</fnm></au>
  </aug>
  <source>Econometrica</source>
  <publisher>Wiley Online Library</publisher>
  <pubdate>2000</pubdate>
  <volume>68</volume>
  <issue>2</issue>
  <fpage>371</fpage>
  <lpage>-398</lpage>
</bibl>

<bibl id="B40">
  <title><p>Conformism and diversity under social learning</p></title>
  <aug>
    <au><snm>Bala</snm><fnm>V</fnm></au>
    <au><snm>Goyal</snm><fnm>S</fnm></au>
  </aug>
  <source>Economic theory</source>
  <publisher>Springer</publisher>
  <pubdate>2001</pubdate>
  <volume>17</volume>
  <issue>1</issue>
  <fpage>101</fpage>
  <lpage>-120</lpage>
</bibl>

<bibl id="B41">
  <title><p>Word-of-mouth learning</p></title>
  <aug>
    <au><snm>Banerjee</snm><fnm>A</fnm></au>
    <au><snm>Fudenberg</snm><fnm>D</fnm></au>
  </aug>
  <source>Games and Economic Behavior</source>
  <publisher>Elsevier</publisher>
  <pubdate>2004</pubdate>
  <volume>46</volume>
  <issue>1</issue>
  <fpage>1</fpage>
  <lpage>-22</lpage>
</bibl>

<bibl id="B42">
  <title><p>Bayesian learning in social networks</p></title>
  <aug>
    <au><snm>Gale</snm><fnm>D</fnm></au>
    <au><snm>Kariv</snm><fnm>S</fnm></au>
  </aug>
  <source>Games and Economic Behavior</source>
  <publisher>Elsevier</publisher>
  <pubdate>2003</pubdate>
  <volume>45</volume>
  <issue>2</issue>
  <fpage>329</fpage>
  <lpage>-346</lpage>
</bibl>

<bibl id="B43">
  <title><p>Sequential decisions with tests</p></title>
  <aug>
    <au><snm>Gill</snm><fnm>D</fnm></au>
    <au><snm>Sgroi</snm><fnm>D</fnm></au>
  </aug>
  <source>Games and economic Behavior</source>
  <publisher>Elsevier</publisher>
  <pubdate>2008</pubdate>
  <volume>63</volume>
  <issue>2</issue>
  <fpage>663</fpage>
  <lpage>-678</lpage>
</bibl>

<bibl id="B44">
  <title><p>Bayesian learning in social networks</p></title>
  <aug>
    <au><snm>Acemoglu</snm><fnm>D</fnm></au>
    <au><snm>Dahleh</snm><fnm>MA</fnm></au>
    <au><snm>Lobel</snm><fnm>I</fnm></au>
    <au><snm>Ozdaglar</snm><fnm>A</fnm></au>
  </aug>
  <source>The Review of Economic Studies</source>
  <publisher>Oxford University Press</publisher>
  <pubdate>2011</pubdate>
  <volume>78</volume>
  <issue>4</issue>
  <fpage>1201</fpage>
  <lpage>-1236</lpage>
</bibl>

<bibl id="B45">
  <title><p>Steering Information Cascades in a Social System by Selective
  Rewiring and Incentive Seeding</p></title>
  <aug>
    <au><snm>Hsiao</snm><fnm>J.</fnm></au>
    <au><snm>Chen</snm><fnm>K. C.</fnm></au>
  </aug>
  <source>to be included in 2016 IEEE International Conference on
  Communications (ICC)</source>
  <pubdate>2016</pubdate>
</bibl>

<bibl id="B46">
  <title><p>Persuasion bias, social influence, and uni-dimensional
  opinions</p></title>
  <aug>
    <au><snm>DeMarzo</snm><fnm>PM</fnm></au>
    <au><snm>Zwiebel</snm><fnm>J</fnm></au>
    <au><snm>Vayanos</snm><fnm>D</fnm></au>
  </aug>
  <source>Social Influence, and Uni-Dimensional Opinions (November 2001). MIT
  Sloan Working Paper</source>
  <pubdate>2001</pubdate>
  <issue>4339-01</issue>
</bibl>

<bibl id="B47">
  <title><p>Naive learning in social networks and the wisdom of
  crowds</p></title>
  <aug>
    <au><snm>Golub</snm><fnm>B</fnm></au>
    <au><snm>Jackson</snm><fnm>MO</fnm></au>
  </aug>
  <source>American Economic Journal: Microeconomics</source>
  <publisher>American Economic Association</publisher>
  <pubdate>2010</pubdate>
  <volume>2</volume>
  <issue>1</issue>
  <fpage>112</fpage>
  <lpage>-149</lpage>
</bibl>

<bibl id="B48">
  <title><p>Spread of (mis) information in social networks</p></title>
  <aug>
    <au><snm>Acemoglu</snm><fnm>D</fnm></au>
    <au><snm>Ozdaglar</snm><fnm>A</fnm></au>
    <au><snm>ParandehGheibi</snm><fnm>A</fnm></au>
  </aug>
  <source>Games and Economic Behavior</source>
  <publisher>Elsevier</publisher>
  <pubdate>2010</pubdate>
  <volume>70</volume>
  <issue>2</issue>
  <fpage>194</fpage>
  <lpage>-227</lpage>
</bibl>

<bibl id="B49">
  <title><p>Non-Bayesian social learning</p></title>
  <aug>
    <au><snm>Jadbabaie</snm><fnm>A</fnm></au>
    <au><snm>Molavi</snm><fnm>P</fnm></au>
    <au><snm>Sandroni</snm><fnm>A</fnm></au>
    <au><snm>Tahbaz Salehi</snm><fnm>A</fnm></au>
  </aug>
  <source>Games and Economic Behavior</source>
  <publisher>Elsevier</publisher>
  <pubdate>2012</pubdate>
  <volume>76</volume>
  <issue>1</issue>
  <fpage>210</fpage>
  <lpage>-225</lpage>
</bibl>

<bibl id="B50">
  <title><p>Social learning and distributed hypothesis testing</p></title>
  <aug>
    <au><snm>Lalitha</snm><fnm>A</fnm></au>
    <au><snm>Sarwate</snm><fnm>A</fnm></au>
    <au><snm>Javidi</snm><fnm>T</fnm></au>
  </aug>
  <source>2014 IEEE International Symposium on Information Theory</source>
  <pubdate>2014</pubdate>
  <fpage>551</fpage>
  <lpage>-555</lpage>
</bibl>

<bibl id="B51">
  <title><p>Distributed hypothesis testing with social learning and symmetric
  fusion</p></title>
  <aug>
    <au><snm>Rhim</snm><fnm>JB</fnm></au>
    <au><snm>Goyal</snm><fnm>VK</fnm></au>
  </aug>
  <source>IEEE Transactions on Signal Processing</source>
  <publisher>IEEE</publisher>
  <pubdate>2014</pubdate>
  <volume>62</volume>
  <issue>23</issue>
  <fpage>6298</fpage>
  <lpage>-6308</lpage>
</bibl>

<bibl id="B52">
  <title><p>Information cascades in social networks via dynamic system
  analyses</p></title>
  <aug>
    <au><snm>Huang</snm><fnm>S. L.</fnm></au>
    <au><snm>Chen</snm><fnm>K. C.</fnm></au>
  </aug>
  <source>2015 IEEE International Conference on Communications (ICC)</source>
  <pubdate>2015</pubdate>
  <fpage>1262</fpage>
  <lpage>1267</lpage>
</bibl>

<bibl id="B53">
  <title><p>Network tomography: recent developments</p></title>
  <aug>
    <au><snm>Castro</snm><fnm>R</fnm></au>
    <au><snm>Coates</snm><fnm>M</fnm></au>
    <au><snm>Liang</snm><fnm>G</fnm></au>
    <au><snm>Nowak</snm><fnm>R</fnm></au>
    <au><snm>Yu</snm><fnm>B</fnm></au>
  </aug>
  <source>Statistical science</source>
  <publisher>JSTOR</publisher>
  <pubdate>2004</pubdate>
  <fpage>499</fpage>
  <lpage>-517</lpage>
</bibl>

<bibl id="B54">
  <title><p>Optimal serial distributed decision fusion</p></title>
  <aug>
    <au><snm>Viswanathan</snm><fnm>R</fnm></au>
    <au><snm>Thomopoulos</snm><fnm>SC</fnm></au>
    <au><snm>Tumuluri</snm><fnm>RAMAKRISHNAJ</fnm></au>
  </aug>
  <source>IEEE Transactions on Aerospace and Electronic Systems</source>
  <publisher>IEEE</publisher>
  <pubdate>1988</pubdate>
  <volume>24</volume>
  <issue>4</issue>
  <fpage>366</fpage>
  <lpage>-376</lpage>
</bibl>

<bibl id="B55">
  <title><p>Distributed detection by a large team of sensors in
  tandem</p></title>
  <aug>
    <au><snm>Papastavrou</snm><fnm>JD</fnm></au>
    <au><snm>Athans</snm><fnm>M</fnm></au>
  </aug>
  <source>IEEE Transactions on Aerospace and Electronic Systems</source>
  <publisher>IEEE</publisher>
  <pubdate>1992</pubdate>
  <volume>28</volume>
  <issue>3</issue>
  <fpage>639</fpage>
  <lpage>-653</lpage>
</bibl>

<bibl id="B56">
  <title><p>On the performance of serial networks in distributed
  detection</p></title>
  <aug>
    <au><snm>Swaszek</snm><fnm>PF</fnm></au>
  </aug>
  <source>IEEE transactions on aerospace and electronic systems</source>
  <publisher>IEEE</publisher>
  <pubdate>1993</pubdate>
  <volume>29</volume>
  <issue>1</issue>
  <fpage>254</fpage>
  <lpage>-260</lpage>
</bibl>

<bibl id="B57">
  <title><p>Serial distributed detection for wireless sensor
  networks</p></title>
  <aug>
    <au><snm>Bahceci</snm><fnm>I</fnm></au>
    <au><snm>Al Regib</snm><fnm>G</fnm></au>
    <au><snm>Altunbasak</snm><fnm>Y</fnm></au>
  </aug>
  <source>Information Theory, 2005. ISIT 2005. Proceedings. International
  Symposium on</source>
  <pubdate>2005</pubdate>
  <fpage>830</fpage>
  <lpage>-834</lpage>
</bibl>

<bibl id="B58">
  <title><p>A Technological Perspective on Information Cascades via Social
  Learning</p></title>
  <aug>
    <au><snm>Rosas</snm><fnm>F</fnm></au>
    <au><snm>Hsiao</snm><fnm>JH</fnm></au>
    <au><snm>Chen</snm><fnm>KC</fnm></au>
  </aug>
  <source>IEEE Access</source>
  <publisher>IEEE</publisher>
  <pubdate>2017</pubdate>
  <volume>5</volume>
  <fpage>22605</fpage>
  <lpage>-22633</lpage>
</bibl>

<bibl id="B59">
  <title><p>Social Learning Against Data Falsification in Sensor
  Networks</p></title>
  <aug>
    <au><snm>Rosas</snm><fnm>F</fnm></au>
    <au><snm>Chen</snm><fnm>KC</fnm></au>
  </aug>
  <source>International Workshop on Complex Networks and their
  Applications</source>
  <pubdate>2017</pubdate>
  <fpage>704</fpage>
  <lpage>-716</lpage>
</bibl>

<bibl id="B60">
  <title><p>Modulation and SNR optimization for achieving energy-efficient
  communications over short-range fading channels</p></title>
  <aug>
    <au><snm>Rosas</snm><fnm>F</fnm></au>
    <au><snm>Oberli</snm><fnm>C</fnm></au>
  </aug>
  <source>IEEE Transactions on Wireless Communications</source>
  <publisher>IEEE</publisher>
  <pubdate>2012</pubdate>
  <volume>11</volume>
  <issue>12</issue>
  <fpage>4286</fpage>
  <lpage>-4295</lpage>
</bibl>

<bibl id="B61">
  <title><p>Applications and trends in wireless acoustic sensor networks: A
  signal processing perspective</p></title>
  <aug>
    <au><snm>Bertrand</snm><fnm>A.</fnm></au>
  </aug>
  <source>2011 18th IEEE Symposium on Communications and Vehicular Technology
  in the Benelux (SCVT)</source>
  <pubdate>2011</pubdate>
  <fpage>1</fpage>
  <lpage>6</lpage>
</bibl>

<bibl id="B62">
  <title><p>Optimal data fusion of correlated local decisions in multiple
  sensor detection systems</p></title>
  <aug>
    <au><snm>Kam</snm><fnm>M</fnm></au>
    <au><snm>Zhu</snm><fnm>Q</fnm></au>
    <au><snm>Gray</snm><fnm>WS</fnm></au>
  </aug>
  <source>IEEE Transactions on Aerospace and Electronic Systems</source>
  <publisher>IEEE</publisher>
  <pubdate>1992</pubdate>
  <volume>28</volume>
  <issue>3</issue>
  <fpage>916</fpage>
  <lpage>-920</lpage>
</bibl>

<bibl id="B63">
  <title><p>Adaptive fusion of correlated local decisions</p></title>
  <aug>
    <au><snm>Chen</snm><fnm>J G</fnm></au>
    <au><snm>Ansari</snm><fnm>N</fnm></au>
  </aug>
  <source>IEEE Transactions on Systems, Man, and Cybernetics, Part C
  (Applications and Reviews)</source>
  <publisher>IEEE</publisher>
  <pubdate>1998</pubdate>
  <volume>28</volume>
  <issue>2</issue>
  <fpage>276</fpage>
  <lpage>-281</lpage>
</bibl>

<bibl id="B64">
  <title><p>The good, bad and ugly: distributed detection of a known signal in
  dependent Gaussian noise</p></title>
  <aug>
    <au><snm>Willett</snm><fnm>P</fnm></au>
    <au><snm>Swaszek</snm><fnm>PF</fnm></au>
    <au><snm>Blum</snm><fnm>RS</fnm></au>
  </aug>
  <source>IEEE Transactions on Signal Processing</source>
  <publisher>IEEE</publisher>
  <pubdate>2000</pubdate>
  <volume>48</volume>
  <issue>12</issue>
  <fpage>3266</fpage>
  <lpage>-3279</lpage>
</bibl>

<bibl id="B65">
  <title><p>How dense should a sensor network be for detection with correlated
  observations?</p></title>
  <aug>
    <au><snm>Chamberland</snm><fnm>J F</fnm></au>
    <au><snm>Veeravalli</snm><fnm>VV</fnm></au>
  </aug>
  <source>IEEE Transactions on Information Theory</source>
  <publisher>IEEE</publisher>
  <pubdate>2006</pubdate>
  <volume>52</volume>
  <issue>11</issue>
  <fpage>5099</fpage>
  <lpage>-5106</lpage>
</bibl>

<bibl id="B66">
  <title><p>Copula-based fusion of correlated decisions</p></title>
  <aug>
    <au><snm>Sundaresan</snm><fnm>A</fnm></au>
    <au><snm>Varshney</snm><fnm>PK</fnm></au>
    <au><snm>Rao</snm><fnm>NS</fnm></au>
  </aug>
  <source>IEEE Transactions on Aerospace and Electronic Systems</source>
  <publisher>IEEE</publisher>
  <pubdate>2011</pubdate>
  <volume>47</volume>
  <issue>1</issue>
  <fpage>454</fpage>
  <lpage>-471</lpage>
</bibl>

<bibl id="B67">
  <title><p>Probability Theory I</p></title>
  <aug>
    <au><snm>Loeve</snm><fnm>M.</fnm></au>
  </aug>
  <publisher>Springer</publisher>
  <pubdate>1978</pubdate>
</bibl>

<bibl id="B68">
  <title><p>Protocols and architectures for wireless sensor
  networks</p></title>
  <aug>
    <au><snm>Karl</snm><fnm>H</fnm></au>
    <au><snm>Willig</snm><fnm>A</fnm></au>
  </aug>
  <publisher>John Wiley \& Sons</publisher>
  <pubdate>2007</pubdate>
</bibl>

<bibl id="B69">
  <title><p>Malware diffusion models for modern complex networks: theory and
  applications</p></title>
  <aug>
    <au><snm>Karyotis</snm><fnm>V</fnm></au>
    <au><snm>Khouzani</snm><fnm>MHR</fnm></au>
  </aug>
  <publisher>Morgan Kaufmann</publisher>
  <pubdate>2016</pubdate>
</bibl>

<bibl id="B70">
  <title><p>An introduction to signal detection and estimation</p></title>
  <aug>
    <au><snm>Poor</snm><fnm>HV</fnm></au>
  </aug>
  <publisher>Springer Science \& Business Media</publisher>
  <pubdate>2013</pubdate>
</bibl>

<bibl id="B71">
  <title><p>Network resilience: a systematic approach</p></title>
  <aug>
    <au><snm>Smith</snm><fnm>P</fnm></au>
    <au><snm>Hutchison</snm><fnm>D</fnm></au>
    <au><snm>Sterbenz</snm><fnm>JP</fnm></au>
    <au><snm>Sch{\"o}ller</snm><fnm>M</fnm></au>
    <au><snm>Fessi</snm><fnm>A</fnm></au>
    <au><snm>Karaliopoulos</snm><fnm>M</fnm></au>
    <au><snm>Lac</snm><fnm>C</fnm></au>
    <au><snm>Plattner</snm><fnm>B</fnm></au>
  </aug>
  <source>IEEE Communications Magazine</source>
  <publisher>IEEE</publisher>
  <pubdate>2011</pubdate>
  <volume>49</volume>
  <issue>7</issue>
  <fpage>88</fpage>
  <lpage>-97</lpage>
</bibl>

<bibl id="B72">
  <title><p>Conversation, information, and herd behavior</p></title>
  <aug>
    <au><snm>Shiller</snm><fnm>RJ</fnm></au>
  </aug>
  <source>The American Economic Review</source>
  <publisher>JSTOR</publisher>
  <pubdate>1995</pubdate>
  <volume>85</volume>
  <issue>2</issue>
  <fpage>181</fpage>
  <lpage>-185</lpage>
</bibl>

<bibl id="B73">
  <title><p>Bayesian data analysis</p></title>
  <aug>
    <au><snm>Gelman</snm><fnm>A</fnm></au>
    <au><snm>Carlin</snm><fnm>JB</fnm></au>
    <au><snm>Stern</snm><fnm>HS</fnm></au>
    <au><snm>Dunson</snm><fnm>DB</fnm></au>
    <au><snm>Vehtari</snm><fnm>A</fnm></au>
    <au><snm>Rubin</snm><fnm>DB</fnm></au>
  </aug>
  <publisher>CRC press Boca Raton, FL</publisher>
  <pubdate>2014</pubdate>
  <volume>2</volume>
</bibl>

<bibl id="B74">
  <title><p>Elements of information theory</p></title>
  <aug>
    <au><snm>Cover</snm><fnm>TM</fnm></au>
    <au><snm>Thomas</snm><fnm>JA</fnm></au>
  </aug>
  <publisher>John Wiley \& Sons</publisher>
  <pubdate>2012</pubdate>
</bibl>

<bibl id="B75">
  <title><p>Understanding interdependency through complex information
  sharing</p></title>
  <aug>
    <au><snm>Rosas</snm><fnm>F</fnm></au>
    <au><snm>Ntranos</snm><fnm>V</fnm></au>
    <au><snm>Ellison</snm><fnm>CJ</fnm></au>
    <au><snm>Pollin</snm><fnm>S</fnm></au>
    <au><snm>Verhelst</snm><fnm>M</fnm></au>
  </aug>
  <source>Entropy</source>
  <publisher>Multidisciplinary Digital Publishing Institute</publisher>
  <pubdate>2016</pubdate>
  <volume>18</volume>
  <issue>2</issue>
  <fpage>38</fpage>
</bibl>

<bibl id="B76">
  <title><p>Treatise on Analysis</p></title>
  <aug>
    <au><snm>Dieudonne</snm><fnm>J.</fnm></au>
  </aug>
  <publisher>Associated Press, New York</publisher>
  <pubdate>1976</pubdate>
  <volume>II</volume>
</bibl>

<bibl id="B77">
  <title><p>Detecting changes in water quality data</p></title>
  <aug>
    <au><snm>McKenna</snm><fnm>SA</fnm></au>
    <au><snm>Wilson</snm><fnm>M</fnm></au>
    <au><snm>Klise</snm><fnm>KA</fnm></au>
  </aug>
  <source>American Water Works Association. Journal</source>
  <publisher>American Water Works Association</publisher>
  <pubdate>2008</pubdate>
  <volume>100</volume>
  <issue>1</issue>
  <fpage>74</fpage>
</bibl>

</refgrp>
}

\end{document}